\documentclass[12pt, a4paper]{amsart}
\usepackage[margin=3cm]{geometry}
\usepackage{amscd, amsmath, amssymb, amsthm, appendix, float, mathtools,  subcaption, tabu, tikz} 
\usetikzlibrary{matrix,arrows,positioning,automata,shapes,shapes.multipart} 

\allowdisplaybreaks

\newtheorem{Theorem}{Theorem}[section]

\newtheorem{Lemma}[Theorem]{Lemma}
\newtheorem{Corollary}[Theorem]{Corollary}
\newtheorem{Conjecture}[Theorem]{Conjecture}
\newtheorem*{Theorem*}{Theorem}

\theoremstyle{definition}
\newtheorem{Definition}[Theorem]{Definition}
\newtheorem*{Definition*}{Definition}

\newtheorem{Example}[Theorem]{Example}

\newcommand{\R}{\mathbb{R}}

\newcommand{\ra}{\rightarrow}

\newcommand{\T}{\mathcal{T}}

\title[Tree-metrizable HGT networks]{Tree-metrizable HGT networks} 

\author{Michael Hendriksen and Andrew Francis}

\address{Centre for Research in Mathematics, Western Sydney University, Australia.}
\email{m.hendriksen@westernsydney.edu.au}
\email{a.francis@westernsydney.edu.au}

\date{\today}

\begin{document}

\begin{abstract}
Phylogenetic trees are often constructed by using a metric on the set of taxa that label the leaves of the tree.  While there are a number of methods for constructing a tree using a given metric, such trees will only display the metric if it satisfies the so-called ``four point condition'', established by Buneman in 1971.  While this condition guarantees that a unique tree will display the metric, meaning that the distance between any two leaves can be found by adding the distances on arcs in the path between the leaves, it doesn't exclude the possibility that a phylogenetic network might also display the metric.  This possibility was recently pointed out and ``tree-metrized'' networks --- that display a tree metric --- with a single reticulation were characterized.  In this paper, we show that in the case of HGT (horizontal gene transfer) networks, in fact there are tree-metrized networks containing many reticulations.
\end{abstract}

\maketitle

\section{Introduction}

Phylogenetic trees have been used to represent the relationships among a set of taxa labelling the leaves since the days of Darwin \cite{Felsenstein2004}.  Especially in the case of trees drawn with a root, the arcs of such rooted trees represent an evolutionary process proceeding over time away from the root and towards the leaves, and vertices in the tree represent divergence, or speciation, events.   Likewise, phylogenetic \emph{networks} have come to prominence recently as a way to represent evolutionary processes in which branches of the tree interact with each other.  Two key examples of such interactions are \emph{hybridization}, in which two species interbreed, resulting in a third, and \emph{horizontal gene transfer}, in which genetic material from one species is acquired by a second \cite{Huson2005}.  

In particular, horizontal gene transfer is highly relevant for studies of evolutionary history --- it is thought to be the primary driver of early cellular evolution \cite{Woese2002}, and still is relevant to ongoing evolution, with over half of total genes in the genomes of human-associated microbiota involved in horizontal gene transfer \cite{Jeong2019}. We will therefore focus in particular on HGT networks throughout this article.

While phylogenetic trees and networks can be constructed in many ways, current approaches often involve a metric on the set of taxa.  That is, a matrix giving the pairwise distances between each pair of leaves of the tree or network.  While such distances are natural to define on a tree, there are different ways one may define the distance between leaves in a network; we will give more details of the approach we take to this, below.  

In this paper we are concerned with metrics on a set of taxa that are able to be placed on a tree --- ``tree metrics'' --- but that can also be placed on a network.  It was recently observed that some tree metrics have this property, and the resulting networks that have a single ``reticulation''  were characterized~\cite{francis2015tree}.  The present paper extends this by investigating networks with \emph{more} than one reticulation that can nevertheless carry tree metrics.  We call such networks ``tree-metrizable''.  

Fortunately, there is an explicit characterization of when a metric can be placed on a tree. The famous ``four-point condition'' (Theorem \ref{t:4PC}), due to Buneman~\cite{buneman1971recovery}, says that if a metric $d$ on a set $X$ satisfies the condition then there exists a unique weighted phylogenetic tree with leaf-set $X$ whose induced metric is $d$.  This, incidentally, provides a characterisation of weighted phylogenetic trees on $X$: a pair of trees are isomorphic as weighted graphs if and only if their induced metrics are identical.   

Metrics aside, reticulation events such as horizontal gene transfer (HGT) and hybridization are not able to be represented on a tree.  For these, various forms of phylogenetic network can be defined, generalizing phylogenetic trees.

In particular, reticulation arcs can be regarded as instantaneous events which carry weights representing the proportion of genetic material carried along them.  This model considers the phylogenetic network to be a linear combination of the trees that it displays, and as such, a metric can be defined from the network by taking a convex linear combination of the metrics corresponding to the displayed trees, and whose coefficients are taken from the weights on the reticulation arcs of the network.   

Surprisingly, it has recently been shown that it is possible for both hybridization networks and HGT networks (defined in Section~\ref{s:HGTnetworks}) to produce metrics that satisfy the four-point condition~\cite{francis2015tree}.  That is, they may carry tree metrics.  The implication is that a metric being a tree metric cannot rule out the evolutionary history of the taxa $X$ being explained by a network.  In fact, \emph{any} tree metric can be displayed by a network (\cite[Theorem 2]{francis2015tree}), although it may be a very simple one.  

A natural question then, is what phylogenetic networks might possibly carry tree metrics?  We call such networks \emph{tree-metrizable networks}.  This question, for (binary) hybridization networks, was answered in \cite{francis2015tree}: the answer was ``not many''.  There are tight restrictions on where hybridizations can occur for the network to carry a tree metric.  The case of HGT networks, however, was left open.  Conditions on networks with a single HGT arc were established, but an example of an HGT network with two HGT arcs was given that carries a tree metric, and what's more, the tree metric corresponded to a tree that was not a base-tree of the network (in the sense of~\cite{francis2015phylogenetic})!

This paper seeks to explore this phenomenon.  
That is, we are interested in the situation in which the inferred metric from a weighted rooted binary HGT network might satisfy the four-point condition, and so be indistinguishable from a tree.  
The key approach in this paper is to graft structures on to the leaf of a tree or network, while maintaining the existence of a metric on the leaves.  With such tools, complicated tree-metrizable networks can be built up from a base tree or network.  

The paper begins with background definitions and results on metrics on trees and HGT networks (Section~\ref{s:background}).  We then begin our exploration of tree-metrizable HGT networks in Section~\ref{s:tree-metrizability}, by first extending the four-point condition to the HGT network context, and then deriving some natural extensions to the results of~\cite{francis2015tree}.  These results effectively show that tree-metrizable networks can be constructed with any number of HGT arcs at all, by adding certain HGT arcs (Lemma~\ref{OmitAdjacent}). 

The final two sections show how complicated tree-metrizable networks can be constructed by grafting trees onto small tree-metrizable networks (Section~\ref{s:network.stock}), and then the reverse (Section~\ref{s:tree.stock}).   The paper concludes with a discussion of the results and some further questions, in Section~\ref{s:discussion}.

\section{Background}\label{s:background} 

\subsection{Trees and tree metrics}

Unless otherwise stated, all trees in this paper are rooted binary phylogenetic $X$-trees, as defined below: 

\begin{Definition}
A \textit{rooted binary phylogenetic $X$-tree} on a set $X$ is a rooted acyclic digraph with the following properties:
\begin{enumerate}
\setlength\itemsep{0em}
\item the \textit{root} vertex has in-degree $0$ and out-degree $2$;
\item $X$ labels the set of vertices with out-degree $0$ and in-degree $1$, called \emph{leaves}; and
\item all remaining vertices have in-degree $1$ and out-degree $2$.
\end{enumerate}
\end{Definition}

The vertices of degree 3 --- all other than the root and the leaves --- are called \emph{internal} vertices.
 We write $V(T)$ for the set of vertices of $T$, and $E(T)$ for the set of arcs.  In a rooted tree, arcs are directed away from the root, and so are ordered pairs of vertices: we will denote the (directed) arc from $u$ to $v$ by $(u,v)$, for $u,v\in V(T)$.  

Metrics, however, are not affected by root placement, and so we will frequently need to consider the \emph{unrooted} tree obtained from the rooted tree by treating arcs as unordered pairs of vertices (referred to as edges to avoid confusion), and ``suppressing'' the root.  By \emph{suppressing} a vertex $v$ of degree 2, we mean deleting the vertex $v$ and the edges incident to it, $\{u_1,v\}$ and $\{u_2,v\}$, and adding a new edge $\{u_1,u_2\}$ (note that an edge in an unrooted tree is simply a pair of distinct vertices).  The reverse operation --- deleting an edge $\{u_1,u_2\}$ and replacing it with a new vertex $v$ and a pair of edges  $\{u_1,v\}$ and $\{u_2,v\}$ --- is called \emph{subdividing} the edge $\{u_1,u_2\}$.  Analogous operations can be defined for directed graphs, so long as when suppressing a vertex of degree two it has indegree 1 and outdegree 1.

We will say two trees $T_1$ and $T_2$ are isomorphic if there is a one-to-one correspondence between their vertices $\phi:V(T_1)\to V(T_2)$ that also maps their edges: $E(T_2)=\{(\phi(u),\phi(v))\mid (u,v)\in E(T_1)\}$ and preserves leaf labels.  Note that if isomorphic trees are rooted, their roots must map to each other.

In particular, there are exactly $3$ isomorphism classes of unrooted trees on a set of taxa $X$ for $|X|=4$ (referred to as a \textit{quartet}). If $X=\{a,b,c,d\}$, and we denote the tree in which the unique paths from $a$ to $b$ and $c$ to $d$ do not intersect by $ab|cd$, then these classes correspond to $ab|cd$, $ac|bd$ and $ad|bc$.

A \emph{weight function} on a rooted binary phylogenetic $X$-tree $T=(V,E)$ is a map $w: E \rightarrow \R^{>0}$ that assigns strictly positive weights to the arcs of a tree. We denote a tree $T$ with associated weight function $w$ by $T^w$. This allows us to define the distance $d(x,y)$ between two leaves $x$ and $y$ in $X$ to be the sum of the weights on the arcs in the unique path between $x$ and $y$. This distance is referred to as the $\textit{tree distance}$ between $x$ and $y$, and any set of pairwise distances between elements of $X$ that can be represented on a tree in this way is referred to as a \textit{tree metric}. If there are two trees $T_1^{w_1}$ and $T_2^{w_2}$ such that $T_1$ and $T_2$ are isomorphic as unrooted trees (but not necessarily $w_1 = w_2$), we denote this as $T_1 \cong T_2$. If $w_1=w_2$ as well, we will refer to $T_1^{w_1}$ and $T_2^{w_2}$ as \textit{isomorphic as weighted trees}, denoted $T_1^{w_1} \cong_w T_2^{w_2}$, or $T_1 \cong_w T_2$ when the corresponding weight functions are clear from context.

A fundamental characterisation of tree metrics is the `four point condition'.

\begin{Theorem}[Four point condition~\cite{buneman1971recovery}]\label{t:4PC}
A set of distances $d$ on a set $X$ is a tree metric on $X$ if and only if for any $x_1,x_2,x_3,x_4 \in X$, two of the three sums 
\[
d(x_1,x_2)+d(x_3,x_4);\quad d(x_1,x_3)+d(x_2,x_4);\quad d(x_1,x_4)+d(x_2,x_3)
\] 
are equal, and are greater than or equal to the third sum.
\end{Theorem}

\subsection{HGT networks}\label{s:HGTnetworks}

A \emph{horizontal gene transfer} (HGT) network is a generalisation of a binary phylogenetic tree that allows the modelling of certain reticulation events.

\begin{Definition}
\label{HGTDef}
A \textit{HGT network} $N$ on a set $X$ is a rooted acyclic digraph $(V,E)$ with the following properties:
\begin{enumerate}
\setlength\itemsep{0em}
\item the \textit{root} vertex has in-degree $0$ and out-degree $2$;
\item $X$ labels the set of vertices with out-degree $0$ and in-degree $1$ (the \emph{leaves});
\item all remaining vertices are \textit{interior vertices} and have either in-degree $1$ and out-degree $2$ (a \textit{tree vertex}), or in-degree $2$ and out-degree $1$ (a \textit{reticulation vertex});
\item the arc set $E$ of $N$ is the disjoint union of two subsets, the set of ‘reticulation arcs’ $E_R$ and the set of ‘tree arcs’ $E_T$; moreover each reticulation arc ends at a reticulation vertex, and each reticulation vertex has exactly one incoming reticulation arc; 
\item every interior vertex has at least one outgoing tree arc; and
\item there is a time function $t : V \rightarrow \R$ so that (a) if $(u,v)$ is a tree arc then $t(u) < t(v)$ and (b) if $(u,v)$ is a reticulation arc, then $t(u) = t(v)$. 
\end{enumerate}
\end{Definition}

Informally, one can think of an HGT network as a binary phylogenetic $X$-tree for which certain arcs are subdivided and a horizontal arc is placed between the subdivisions. Throughout this paper, unless otherwise stated, all networks are HGT networks. 

Given an HGT network $N$ on $X$, suppose that for each reticulation vertex we delete exactly one of the incoming arcs. If we then delete any unlabelled leaves formed by this process, the resulting graph is a rooted tree on $X$ whose root is the root of $N$. If all of the degree $2$ vertices aside from the root are then suppressed, the resulting graph is a rooted {binary} phylogenetic $X$-tree, $T$. We say that $T$ is displayed by $N$, and $\T_N$ denotes the set of trees displayed by $N$. See Figure \ref{f:unlabelled.leaves} for an example of this process.

\begin{figure}[H]
\centering
\begin{tikzpicture}[xscale=.3,yscale=0.4]
\draw (0,0) --(7,7)--(14,0);
\draw (3,3) --(6,0);
\draw (11,3) --(8,0);

\draw [>=latex, ->, dashed, thick] (4,2) --(10,2); 
\draw [>=latex, ->, dashed, thick] (1,1) --(9,1);

\draw[fill] (3,3) circle [radius=1.5pt];
\draw[fill] (11,3) circle [radius=1.5pt];
\draw[fill] (4,2) circle [radius=1.5pt];
\draw[fill] (10,2) circle [radius=1.5pt];
\draw[fill] (1,1) circle [radius=1.5pt];
\draw[fill] (9,1) circle [radius=1.5pt];

\node[below] at (0,0) {$\mathstrut x_1$};
\node[below] at (6,0) {$\mathstrut x_2$};
\node[below] at (8,0) {$\mathstrut x_3$};
\node[below] at (14,0) {$\mathstrut x_4$};

\node[above] at (7,2) {$\mathstrut a_1$};
\node[below] at (3,1) {$\mathstrut a_2$};

\node[below] at (7,-2) {$\mathstrut (i)$};
\end{tikzpicture}
\begin{tikzpicture}[xscale=.3,yscale=0.4]
\draw (0,0) --(7,7)--(14,0);
\draw (3,3) --(6,0);
\draw (9,1) --(8,0);

\draw [>=latex, ->, dashed, thick] (4,2) --(10,2); 
\draw [>=latex, ->, dashed, thick] (1,1) --(9,1);

\draw[fill] (3,3) circle [radius=1.5pt];
\draw[fill] (11,3) circle [radius=1.5pt];
\draw[fill] (4,2) circle [radius=1.5pt];
\draw[fill] (10,2) circle [radius=1.5pt];
\draw[fill] (1,1) circle [radius=1.5pt];
\draw[fill] (9,1) circle [radius=1.5pt];

\node[below] at (0,0) {$\mathstrut x_1$};
\node[below] at (6,0) {$\mathstrut x_2$};
\node[below] at (8,0) {$\mathstrut x_3$};
\node[below] at (14,0) {$\mathstrut x_4$};

\node[above] at (7,2) {$\mathstrut a_1$};
\node[below] at (3,1) {$\mathstrut a_2$};

\node[below] at (7,-2) {$\mathstrut (ii)$};
\end{tikzpicture}
\begin{tikzpicture}[xscale=.3,yscale=0.4]
\draw (0,0) --(7,7)--(14,0);
\draw (3,3) --(6,0);
\draw (1,1) --(2,0);

\draw[fill] (3,3) circle [radius=1.5pt];
\draw[fill] (1,1) circle [radius=1.5pt];

\node[below] at (0,0) {$\mathstrut x_1$};
\node[below] at (6,0) {$\mathstrut x_2$};
\node[below] at (2,0) {$\mathstrut x_3$};
\node[below] at (14,0) {$\mathstrut x_4$};

\node[below] at (7,-2) {$\mathstrut (iii)$};
\end{tikzpicture}
\caption{(i) an HGT network $N$ with HGT arcs $a_1,a_2$. Denote the other parent arcs of the reticulation vertices by $a_1',a_2'$ respectively; (ii) The resulting graph after deleting $a_1',a_2'$; (iii) The resulting display tree after deletion of unlabelled leaves and suppression of degree $2$ nodes.}
\label{f:unlabelled.leaves}
\end{figure}
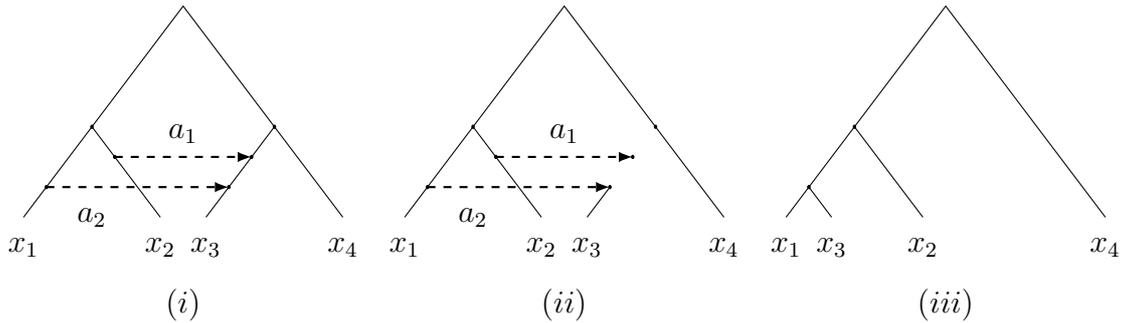

HGT networks have the particularly useful property of having a `canonical' display tree, obtained by deleting all of the reticulation arcs. This tree is referred to as the \textit{underlying tree} of $N$, and is denoted $T_N$.  Note that the underlying tree of an HGT network is a \emph{base} tree in the sense of~\cite{francis2015phylogenetic}, but is not necessarily the only base tree of the network.

\begin{figure}[H]
\centering
\begin{tikzpicture}[scale=0.3]
\draw (0,0) --(5,5)--(10,0);
\draw (5,5) --(7,7)--(14,0);
\draw (9,1) --(8,0);
\draw (2,2) --(4,0);
\draw (8,2) --(6,0);
\draw [>=latex, ->, dashed, thick] (1.5,1.5) --(7.5,1.5);
\draw [>=latex, <-, dashed, thick] (6,4) --(10,4);
\draw [>=latex, ->, dashed, thick] (3.5,0.5) --(6.5,0.5);
\draw [>=latex, ->, dashed, thick] (8.5,1.5) --(12.5,1.5);

\draw[fill] (2,2) circle [radius=1.5pt];
\draw[fill] (8,2) circle [radius=1.5pt];
\draw[fill] (1.5,1.5) circle [radius=1.5pt];
\draw[fill] (7.5,1.5) circle [radius=1.5pt];
\draw[fill] (3.5,0.5) circle [radius=1.5pt];
\draw[fill] (6.5,0.5) circle [radius=1.5pt];
\draw[fill] (6,4) circle [radius=1.5pt];
\draw[fill] (10,4) circle [radius=1.5pt];
\draw[fill] (8.5,1.5) circle [radius=1.5pt];
\draw[fill] (12.5,1.5) circle [radius=1.5pt];
\draw[fill] (5,5) circle [radius=1.5pt];
\draw[fill] (9,1) circle [radius=1.5pt];

\node[below] at (0,0) {$\mathstrut x_1$};
\node[below] at (4,0) {$\mathstrut x_2$};
\node[below] at (6,0) {$\mathstrut x_3$};
\node[below] at (8,0) {$\mathstrut x_4$};
\node[below] at (10,0) {$\mathstrut x_5$};
\node[below] at (14,0) {$\mathstrut x_6$};
\node at (7,-4) {$\mathstrut{(i)}$};
\end{tikzpicture}
\begin{tikzpicture}[scale=0.3]
\draw (0,0) --(5,5)--(10,0);
\draw (5,5) --(7,7)--(14,0);
\draw (9,1) --(8,0);
\draw (2,2) --(4,0);
\draw (8,2) --(6,0);

\draw[fill] (2,2) circle [radius=1.5pt];
\draw[fill] (8,2) circle [radius=1.5pt];
\draw[fill] (5,5) circle [radius=1.5pt];
\draw[fill] (9,1) circle [radius=1.5pt];

\node[below] at (0,0) {$\mathstrut x_1$};
\node[below] at (4,0) {$\mathstrut x_2$};
\node[below] at (6,0) {$\mathstrut x_3$};
\node[below] at (8,0) {$\mathstrut x_4$};
\node[below] at (10,0) {$\mathstrut x_5$};
\node[below] at (14,0) {$\mathstrut x_6$};
\node at (7,-4) {$\mathstrut{(ii)}$};
\end{tikzpicture}
\caption{(i) an HGT network $N$, with reticulation arcs shown dashed; (ii) the underlying tree $T_N$ of $N$.}
\end{figure}
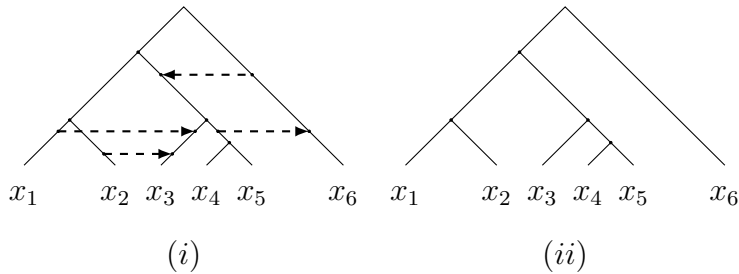

\subsection{HGT network distances}

Following~\cite{francis2015tree}, we define distances on a network $N$ by treating it as a weighted union of the set of $X$-trees obtained by making choices at each reticulation.  For each vertex $v$ in the set $V_R$ of reticulation vertices of $N$, let $R(v)$ denote the two arcs that end at $v$. 
We write $N^w$ for a network $N$ with $w$ a weight function on the tree arcs, $w : E_T \rightarrow \R^{>0}$ and let $\beta$ be a strictly positive probability distribution on the set $F_N$ of functions $f : V_R \rightarrow E$ for which $f(v) \in R(v)$. Each function $f$ describes a weighted tree $T_f$ with induced weight function $w_f$, by specifying a parent for each reticulation vertex. This function $f$, and its tree $T_f$, can then be given an associated probability $\beta_f$, which we construct as follows. 

For each reticulation vertex $v$ with incoming arcs $R(v)=\{a,a'\}$, we associate a function $\alpha$ that gives a number between 0 and 1 to each such reticulation arc,  $\alpha : R(v) \rightarrow (0,1)$, that satisfies $\alpha(a)+\alpha(a') = 1$.  We refer to $\alpha(a)$ and $\alpha(a')$ as the \textit{reticulation probabilities} of $a$ and $a'$ respectively. 
Then let
\begin{equation} 
\beta_f = \prod_{v \in V_R} \alpha(f(v)).
\label{eq:beta.def}
\end{equation}
That is, $\beta_f$ is the product of the weights on the arcs chosen by $f$ for each reticulation vertex.

A distance function 
\[d = d_{\left(N^w,\beta\right)} : X \times X \rightarrow \R^{\ge 0}\]
on $X$ can then be defined for $N^w$ by setting
\begin{equation} d(x,y) = \sum_{f \in F_N} \beta_f d_{\left(T_f^{w_f}\right)} (x,y),
\label{eq:dist.on.N}
\end{equation}
where $w_f$ is the weight function induced by $N^w$ on $T_f$. If there are no reticulation vertices in $N^w$, $d$ is the tree metric $d_{T^w}$. As noted in~\cite[\S 2.3]{francis2015tree}, since $d_{(N^w,\beta)}$ is a convex combination of metrics on $X$, $d$ is also a metric. 

We denote the set of weighted trees obtained in this way from $N^w$ by $\T_N^w$.

The probability distribution on functions $f$ naturally corresponds to a probability distribution on the associated trees $T^w_f\in \T_N^w$, by setting $\beta(T^w_f)=\beta(f)$.  We will drop the reference to $w$ and $f$ where this is not explicitly needed, writing $\beta(T)$.

\begin{Example}
Consider the  weighted network $N^w$ with horizontal reticulation arcs $A,B,C$ shown in Figure~\ref{DisplayExample}. We intend to calculate the associated probability $\beta_f$ of a particular weighted display tree. We have omitted the weights from the diagram for ease of interpretation, but note that we are calculating $\beta_f$ for a particular \textit{weighted} display tree.
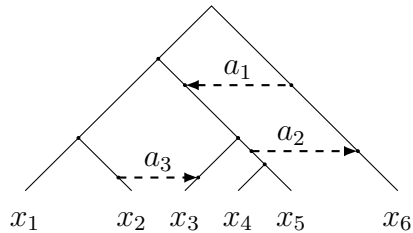
\begin{figure}[H]
\centering
\begin{tikzpicture}[scale=0.35]
\draw (0,0) --(5,5)--(10,0);
\draw (5,5) --(7,7)--(14,0);
\draw (9,1) --(8,0);
\draw (2,2) --(4,0);
\draw (8,2) --(6,0);
\draw [>=latex,<-, dashed, thick] (6,4) --(10,4);
\draw [>=latex, ->, dashed, thick] (3.5,0.5) --(6.5,0.5);
\draw [>=latex, ->, dashed, thick] (8.5,1.5) --(12.5,1.5);

\draw[fill] (2,2) circle [radius=1.5pt];
\draw[fill] (8,2) circle [radius=1.5pt];
\draw[fill] (3.5,0.5) circle [radius=1.5pt];
\draw[fill] (6.5,0.5) circle [radius=1.5pt];
\draw[fill] (6,4) circle [radius=1.5pt];
\draw[fill] (10,4) circle [radius=1.5pt];
\draw[fill] (8.5,1.5) circle [radius=1.5pt];
\draw[fill] (12.5,1.5) circle [radius=1.5pt];
\draw[fill] (5,5) circle [radius=1.5pt];
\draw[fill] (9,1) circle [radius=1.5pt];

\node[below] at (0,0) {$\mathstrut x_1$};
\node[below] at (4,0) {$\mathstrut x_2$};
\node[below] at (6,0) {$\mathstrut x_3$};
\node[below] at (8,0) {$\mathstrut x_4$};
\node[below] at (10,0) {$\mathstrut x_5$};
\node[below] at (14,0) {$\mathstrut x_6$};
\node[above=-0.1cm] at (5,0.5) {$\mathstrut a_3$};
\node[above=-0.1cm] at (10,1.5) {$\mathstrut a_2$};
\node[above=-0.1cm] at (8,4) {$\mathstrut a_1$}; 
\end{tikzpicture}
\caption{A weighted network $N^w$ (with weights omitted) on $6$ leaves.}
\label{DisplayExample}
\end{figure}

Denote the second arc ending at the same vertex as $a_1,a_2$ and $a_3$ respectively by $a_1', a_2'$ and $a_3'$. Then by making the selection $a_1,a_2',a_3$ (and thus deleting $a_1',a_2$ and $a_3'$), we obtain the following display tree $T^w$.
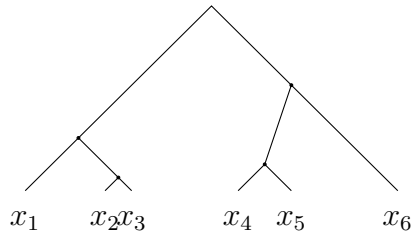
\begin{figure}[H]
\centering
\begin{tikzpicture}[scale=0.35]
\draw (9,1)--(10,0);
\draw (0,0) --(7,7)--(14,0);
\draw (9,1) --(8,0);
\draw (2,2) --(4,0);
\draw (3.5,0.5) --(3,0);
\draw (9,1) --(10,4);

\draw[fill] (2,2) circle [radius=1.5pt];
\draw[fill] (3.5,0.5) circle [radius=1.5pt];
\draw[fill] (10,4) circle [radius=1.5pt];
\draw[fill] (9,1) circle [radius=1.5pt];

\node[below] at (0,0) {$\mathstrut x_1$};
\node[below] at (3,0) {$\mathstrut x_2$};
\node[below] at (4,0) {$\mathstrut x_3$};
\node[below] at (8,0) {$\mathstrut x_4$};
\node[below] at (10,0) {$\mathstrut x_5$};
\node[below] at (14,0) {$\mathstrut x_6$};
\end{tikzpicture}
\caption{The weighted display tree $T^w$ (with weights omitted) obtained from $N^w$ in Figure \ref{DisplayExample} by deleting $a_1',a_2$ and $a_3'$.}
\end{figure}

Then if the reticulation probabilities are $\alpha(a_1)=0.6, \alpha(a_2)=0.2, \alpha(a_3)=0.1$, then the probability assigned to $T^w$ is
\begin{align*}
\beta(T) & = \alpha(a_1)\alpha(a_2')\alpha(a_3) \\
& = \alpha(a_1)(1-\alpha(a_2))\alpha(a_3) \\
& = 0.6 \times 0.8 \times 0.1 \\
& = 0.048.
\end{align*} 
\end{Example}
Somewhat surprisingly, it has recently been shown that HGT networks under this weighted average distance model can obey the four-point condition~\cite{francis2015tree}. That is, the distances represented by some HGT networks can also be represented on a unique tree.  We call such networks ``tree-metrizable'':

\begin{Definition}
Let $N$ be an HGT network on $X$. If there exist arc weights and reticulation probabilities that can be placed on $N$ so that $d_N$ is a tree metric that can be placed on some unweighted tree $T$, we say that $N$ is \textit{tree-metrizable}, or specifically \textit{$T$-metrizable}.
\end{Definition}

We now provide some of the background results on tree-metrizable networks, with wording changed to use the language of tree-metrizability. Throughout this paper, we say that two arcs are adjacent if they are adjacent in the unrooted tree obtained by suppressing the root vertex.

\begin{Lemma}[\cite{francis2015tree}, Lemma 4]
\label{AdjacentLemma}
For any unweighted HGT network $N$, if each reticulation arc is between adjacent tree arcs of $T_N$, then $N$ is $T$-metrizable if and only if $T \cong T_N$.
\end{Lemma}

It easily follows as a side note that any weighted network on $3$ leaves is $T$-metrizable for $T$ any $3$-leaf tree, as all arcs in a $3$-leaf tree are adjacent.

In light of this result, we want to focus on networks that can potentially represent tree metrics that are \emph{not} the underlying tree. 

\begin{Definition}
Let $N$ be an HGT network with reticulation arc $A$. If $A$ is between two non-adjacent arcs of the underlying tree, we say that $A$ is a \textit{non-trivial} reticulation arc. If $N$ contains at least one non-trivial reticulation arc, then $N$ is said to be a \textit{non-trivial} HGT network.
\end{Definition}

Of course, this distinction would not be very useful in our context if there were no networks that were tree-metrizable on a tree that was not the underlying tree.  The proof of the following theorem involves constructing such a network on four leaves with two reticulations.

\begin{Theorem}[\cite{francis2015tree}, Theorem 5(b)]
\label{QuartetExample}
There exist 2-reticulated HGT networks $N$ that are $T_N$-metrizable and (for other parameter settings) $T$-metrizable for $T \not \cong T_N$, even when the mixing distribution treats the two reticulations independently.
\end{Theorem}

The following simple (yet surprisingly powerful) result will be our primary tool for showing a network is not tree-metrizable throughout this paper.

\begin{Lemma}[\cite{francis2015tree}, Lemma 6]
\label{NotTwo}
Let $N$ be a network with display trees $\T_N=\{T_1,\dots,T_k\}$. Suppose that there is a quartet $q \subseteq X$, for which $\big| \{T_i\vert_q \}_{i=1,\dots,k}\big|=2$. Then $N$ is not tree-metrizable.
\end{Lemma}

An immediate consequence of this lemma is that if a network contains a single non-trivial reticulation arc, it is certainly not tree-metrizable, as it will have exactly two non-isomorphic display trees (in both the rooted and unrooted senses). One can then find a quartet upon which the two trees do not agree and then apply the theorem.

\section{Tree-metrizability: first results}\label{s:tree-metrizability}

The following lemma will make our calculations involving the four-point condition easier by  phrasing the four-point condition in terms of the lengths of the internal arcs of a quartet, instead of the tree distances between leaves. 
We defer the proof to the Appendix.

\begin{Definition}
For $T$ a rooted tree, let $T^U$ be the unrooted tree obtained by suppressing the root vertex. That is, if $r$ is the root vertex, we delete the vertex $r$ and edges $(r,u)$ and $(r,v)$, then add $(u,v)$. All edges are then interpreted as undirected.
\end{Definition}

 \begin{Lemma}
 \label{InequalityLemma}
 Let $N$ be a four-leaf network with exactly three display trees, $\T_N = \{T_r, T_s, T_t \}$, and let $\{T_1^{w_1},\dots,T_k^{w_k} \}$ be the weighted trees obtained from $N$ by choices of reticulations. Let $\alpha_j$ be the probability assigned to $T_j$, and $p_j$ be the length of the internal arc of $T_j^U$. 

 Then $N$ is tree-metrizable on $T_r$ if and only if
 \begin{center}
 $\displaystyle \sum_{T_j \cong T_r} p_j \alpha_j > \sum_{T_j \cong T_s} p_j \alpha_j = \sum_{T_j \cong T_t} p_j \alpha_j$.
 \end{center}
 \end{Lemma}

\begin{Example}
Lemma \ref{InequalityLemma} somewhat surprisingly reveals that tree-metrizability is dependant only on the internal arcs of the display trees. For example, let $N$ be the network shown in Figure \ref{InternalUse}. Then,  in order for $N$ to be $T_1$-metrizable, for example, by Lemma \ref{InequalityLemma} we require that 
\[(1-\alpha_1)(1-\alpha_2)(a_1+a_4) \ge \alpha_1(1-\alpha_2)a_3 = \alpha_2a_1.\]  

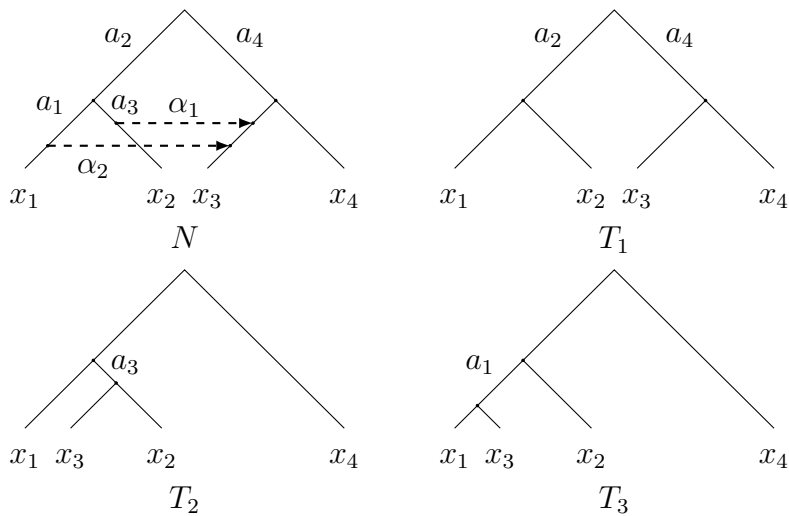
\begin{figure}[h]
\centering
\begin{tikzpicture}[scale=0.3]
\draw (0,0) --(7,7)--(14,0);
\draw (3,3) --(6,0);
\draw (11,3) --(8,0);

\draw [>=latex, ->, dashed, thick] (4,2) --(10,2); 
\draw [>=latex, ->, dashed, thick] (1,1) -- (9,1);

\draw[fill] (3,3) circle [radius=1.5pt];
\draw[fill] (11,3) circle [radius=1.5pt];
\draw[fill] (4,2) circle [radius=1.5pt];
\draw[fill] (10,2) circle [radius=1.5pt];
\draw[fill] (1,1) circle [radius=1.5pt];
\draw[fill] (9,1) circle [radius=1.5pt];

\node[below] at (0,0) {$\mathstrut x_1$};
\node[below] at (6,0) {$\mathstrut x_2$};
\node[below] at (8,0) {$\mathstrut x_3$};
\node[below] at (14,0) {$\mathstrut x_4$};

\node[above=-1mm] at (7,2) {$\mathstrut \alpha_1$};
\node[below=-1mm] at (3,1) {$\mathstrut \alpha_2$};

\node[above left=-0.1cm] at (2,2) {$\mathstrut a_1$};
\node[above left=-0.1cm] at (5,5) {$\mathstrut a_2$};
\node[above right=-0.2cm] at (3.75,2.25) {$\mathstrut a_3$};
\node[above right=-0.1cm] at (9,5) {$\mathstrut a_4$};

\node[below] at (7,-2) {$\mathstrut N$};
\end{tikzpicture}
\hspace{0.5cm}
\begin{tikzpicture}[scale=0.3]
\draw (0,0) --(7,7)--(14,0);
\draw (3,3) --(6,0);
\draw (11,3) --(8,0);

\draw[fill] (3,3) circle [radius=1.5pt];
\draw[fill] (11,3) circle [radius=1.5pt];

\node[below] at (0,0) {$\mathstrut x_1$};
\node[below] at (6,0) {$\mathstrut x_2$};
\node[below] at (8,0) {$\mathstrut x_3$};
\node[below] at (14,0) {$\mathstrut x_4$};

\node[above left=-0.1cm] at (5,5) {$\mathstrut a_2$};
\node[above right=-0.1cm] at (9,5) {$\mathstrut a_4$};

\node[below] at (7,-2) {$\mathstrut T_1$};
\end{tikzpicture}

\begin{tikzpicture}[scale=0.3]
\draw (0,0) --(7,7)--(14,0);
\draw (3,3) --(6,0);
\draw (4,2) --(2,0);

\draw[fill] (3,3) circle [radius=1.5pt];
\draw[fill] (4,2) circle [radius=1.5pt];

\node[below] at (0,0) {$\mathstrut x_1$};
\node[below] at (6,0) {$\mathstrut x_2$};
\node[below] at (2,0) {$\mathstrut x_3$};
\node[below] at (14,0) {$\mathstrut x_4$};

\node[above right=-0.2cm] at (3.75,2.25) {$\mathstrut a_3$};

\node[below] at (7,-2) {$\mathstrut T_2$};
\end{tikzpicture}
\hspace{0.5cm}
\begin{tikzpicture}[scale=0.3]
\draw (0,0) --(7,7)--(14,0);
\draw (3,3) --(6,0);
\draw (1,1) --(2,0);

\draw[fill] (3,3) circle [radius=1.5pt];
\draw[fill] (1,1) circle [radius=1.5pt];

\node[below] at (0,0) {$\mathstrut x_1$};
\node[below] at (6,0) {$\mathstrut x_2$};
\node[below] at (2,0) {$\mathstrut x_3$};
\node[below] at (14,0) {$\mathstrut x_4$};

\node[above left=-0.1cm] at (2,2) {$\mathstrut a_1$};

\node[below] at (7,-2) {$\mathstrut T_3$};
\end{tikzpicture}
\caption{A network $N$ with its three display trees, $T_1, T_2$ and $T_3$.}
\label{InternalUse}
\end{figure}

\end{Example}

The next lemma shows that certain substructures in a weighted network can be interchanged without changing the metric on the leaves. The proof is also deferred to the Appendix.

Recall that Lemma \ref{AdjacentLemma} says that if all of the HGT arcs in an HGT network are between adjacent tree arcs, then the network is tree-metrizable.  In Lemma~\ref{OmitAdjacent}, we generalise this to show how an HGT arc can be placed between any adjacent pair of tree arcs in a network, and retain tree distances (and hence tree-metrizability, if our network is tree-metrizable). In particular, this implies that tree-like distances are preserved with the addition of an arc between adjacent arcs, even if the network contains other arcs between non-adjacent arcs.
 
\begin{Lemma}
\label{OmitAdjacent}
Let $N^w$ be an HGT network with reticulation probabilities $\beta$, and an HGT arc $a$ between a pair of siblings such that there is no other vertex between the ends of $a$ and their common parent vertex. Let $\widehat{N}$ be the HGT network obtained by deleting $a$ from $N$ and suppressing the vertices at each end. Then there exist arc weights $\widehat{w}$ and reticulation probabilities $\widehat{\beta}$ on $N_1$ so that
\begin{center}
$d_{(\widehat{N}^{\widehat{w}},\widehat{\beta})} = d_{(N^w,\beta)}$.
\end{center}
\end{Lemma}

\section{Leaf-Grafting}\label{s:network.stock}

Theorem \ref{QuartetExample} provides an example of a network on four leaves that has two non-trivial reticulation arcs but still is tree-metrizable. 
In the previous section, we addressed the question of whether tree-metrizable networks with more reticulation arcs exist; in this section we address the analogous question for the number of leaves.  In particular, we will show how ``leaf-grafting'' trees onto the leaves of a tree-metrizable network can create a non-trivial tree-metrizable network on any base tree at all.

We begin by defining the notion of \emph{leaf-grafting}.

\begin{Definition}
Let $N$ and $N'$ be two HGT networks on $X$ and $Y$ respectively, and $a$ a leaf of $N$. If we identify the root of $N'$ with the leaf $x$ of $N$, the resulting network $N \#_x N'$ on $X \cup Y - \{a\}$ is termed a \textit{leaf-graft of $N'$ onto $N$ at leaf $x$}. In particular, $N$ is referred to as the \textit{stock}, $N'$ is referred to as the \textit{scion} and $x$ as the \textit{grafting vertex}. 
\end{Definition}

To begin with, we address the case where the stock is an HGT network and the scion is a tree, that is, grafting a tree onto a network. We will consider the reverse problem later. First, a relevant definition.

\begin{Definition}[\cite{Fischer2010}, Section 3]
Let $N$ be an HGT network and $L \subseteq V(N)$ . Then the \textit{lowest stable ancestor} of $L$, denoted $LSA(L)$ is the lowest vertex that lies on every path from the root to a vertex in $L$.
\end{Definition}

We note that the lowest stable ancestor must always exist, as at the very least the root will fit the criterion.

\begin{Theorem}[Replacement Theorem]
\label{t:LeafReplacement}
Let $N$ be a tree-metrizable network, $T$ be a tree, and $\ell$ a leaf of $N$. 
Then $N \#_\ell T$ is a tree-metrizable network if and only if $N$ is a tree-metrizable network.  
\end{Theorem}

\begin{proof}
Suppose $N$ is tree-metrizable. Consider a quartet of leaves $q=\{x_1,x_2,x_3,x_4 \}$ of $N \#_\ell T$.  We will show that the distances between these leaves satisfy the inequality in the four point condition.

Let the leaves of $T$ be denoted by $Y$ and the leaves of $N$ denoted by $X$, so that the leaf set of $N \#_\ell T$ is $\left( X \backslash \{\ell\} \right) \cup Y$. We consider cases according to how many of the leaves of $q$ are in $Y$.

Case (4): All four leaves  of $q$ are in $Y$.  In this case the distances between the leaves are determined by their distances in $T$, and so satisfy the four point condition.

Case (3): Three leaves (say $x_1,x_2,x_3$) are in $Y$ while the fourth, $x_4$, is in $X$.  Let $x=LSA\{x_1,x_2,x_3\}$, and suppose without loss of generality that $x_1x_2|x_3$ forms a rooted triple with root $\rho$.  Then in the network $N \#_\ell T$, the distances between $x_1,x_2,x_3$ and $\rho$ are all determined by the tree $T$.  All distances between $x_1,x_2$ or $x_3$ and $x_4$ go through $\rho$, so that for instance $d(x_1,x_4)=d(x_1,\rho)+d(\rho,x_4)$.  It is immediate that the inequality holds in this case (namely $d(x_1,x_2)+d(x_3,x_4) \le d(x_1,x_3)+d(x_2,x_4)=d(x_1,x_4)+d(x_2,x_3)$).

Case (2): Suppose $x_1,x_2$ are leaves of $Y$ and $x_3,x_4$ are leaves of $X$.  
Then 
\begin{equation}\label{eq:distance.thru.N}
d(x_1,x_3)=d_T(x_1,\ell)+d_N(\ell,x_3),
\end{equation}
and likewise for other cross-pairs $d(x_1,x_4), d(x_2,x_3)$ and $d(x_2,x_4)$.  It is easy to check that $d(x_1,x_3)+d(x_2,x_4)=d(x_1,x_4)+d(x_2,x_3)$.  The inequality $d(x_1,x_2)+d(x_3,x_4)<d(x_1,x_3)+d(x_2,x_4)$ also follows using the observation that $d_T(x_1,x_2)<d_T(x_1,\ell)+d_T(x_2,\ell)$ (and similarly for $d(x_3,x_4)$), by the triangle inequality.

Case (1): Suppose that $x_1$ is in $Y$ and $x_2.x_3,x_4$ are in $X$.  Then the pairwise distances between $x_2,x_3$ and $x_4$ are determined by their distances in $N$, while the distance $d(x_1,x_2)= d_T(x_1,\ell)+d_N(x_2,\ell)$ and similarly for the distances from $x_1$ to $x_3$ and $x_4$. Then as $\ell$ was a leaf of a tree-metrizable network, it follows that the pairwise distances obey the four-point condition.

Case (0): If all leaves in $q$ are in $X$, then their pairwise distances are determined by their distances in $N$, and so satisfy the four-point condition as $N$ is tree-metrizable. 

It follows that $N \#_\ell T$ is tree-metrizable.

Now suppose that $N$ is not tree-metrizable. It follows that there exists a quartet of leaves $q'=\{x_1,x_2,x_3,x_4\}$ of $N$ that does not obey the four-point condition. If $q'$ does not contain $\ell$, then the same quartet in $N \#_\ell T$ does not obey the four-point condition, and therefore $N \#_\ell T$ is not tree-metrizable.

If $q'$ does contain $\ell$, suppose $q'=\{x_1,x_2,x_3,\ell \}$. Then we can select some leaf $k$ of $T$, and observe that $d(x_1,k)=d_N(x_1,\ell)+d(\ell,k)$, with similar forms for $x_2$ and $x_3$. We now consider the distances arising from the quartet $\{x_1,x_2,x_3,k\}$:
\[
d(x_1,x_2)+d(x_3,k);\quad d(x_1,k)+d(x_2,x_3);\quad d(x_1,x_3)+d(x_2,k).
\] 
We can see that 
\begin{align*}
d(x_1,x_2)+d(x_3,k)&= d_N(x_1,x_2)+d_N(x_3,\ell)+d(\ell,k),\\
d(x_1,k)+d(x_2,x_3)&= d_N(x_1,\ell)+d_N(x_2,x_3)+d(\ell,k),\\
d(x_1,x_3)+d(x_2,k)&= d_N(x_1,x_3)+d_N(x_2,\ell)+d(\ell,l).
\end{align*}
 In particular, these are just the corresponding distances of the four-point condition applied to $q'=\{x_1,x_2,x_3,\ell \}$, each with the same distance $d(\ell,k)$ added. It follows that if $q'$ does not obey the four-point condition, neither does $\{x_1,x_2,x_3,k\}$. Therefore $N \#_\ell T$ is not tree-metrizable.
\end{proof}

\begin{Corollary}
There exist non-trivial level-$2$ tree-metrizable networks on $n$ leaves for $n \ge 4$. Furthermore, there exist networks $N$ on $n$ leaves that are $T$-metrizable for some tree $T$ that is not the underlying tree $T_N$.
\end{Corollary}

\begin{proof}
Theorem \ref{QuartetExample} proves this for $n=4$.

If $n > 4$, simply take the level-$2$ network $N$ described in Theorem \ref{QuartetExample} and leaf-graft some tree $T$ with $n -3$ leaves onto any leaf of $N$. Theorem \ref{QuartetExample} also provides an example where a quartet represents a tree that is not its underlying tree. Using this example as our network $N$ provides an example for the second part of this result.
\end{proof}

\section{Caterpillar Networks}

Leaf-grafting provides a neat method for constructing tree-metrizable networks on an arbitrary number of leaves with interesting properties. We will now define a class of tree-metrizable networks on $n$ leaves with $n-2$ non-trivial reticulations, referred to as \textit{caterpillar networks}. In combination with leaf-grafting, this result shows that any tree $T$ of height $h$ can be represented on a tree-metrizable network with $h-1$ reticulation arcs.

We will require the following standard definition.

\begin{Definition}
Let $T$ be a tree, and suppose the leaves $x_1$ and $x_2$ are both children of the same vertex $v$. Then we say that $x_1$ and $x_2$ form a \textit{cherry}, and denote it $\widehat{x_1x_2}$.
\end{Definition}

We can now define caterpillar networks.

\begin{Definition}
Let $C$ be a network with a caterpillar underlying tree $T$ on $n>3$ leaves. Let $C$ be depicted with each internal tree vertex the left child of its parent vertex, and label the leaves $x_1,\dots,x_n$ from left to right, so that $\widehat{x_1x_2}$ is the unique cherry in $T$. For $1\le i <n-1$, let each leaf $x_i$ have a reticulation arc extending to leaf $x_n$, so that the arcs are attached to $x_n$ in numerical order from bottom to top. Then $C$ is referred to as a \textit{caterpillar network}.

For a caterpillar network $C$ on $n$ leaves, let $T_i$ ($1\le i\le n-2$) denote the unique display tree containing the cherry $\widehat{x_ix_n}$, and let $T_{n-1}$ be the underlying tree of $C$.  This uniquely defines $T_i$ because each display tree of $N$ can only contain at most one of the reticulation arcs added to $T$, because they all end on the same tree edge between the root and leaf $x_n$.  Thus $C$ contains exactly $n-1$ display trees (so that $\T_C=\{T_1,\dots,T_{n-1}\}$), with $T_i$ ($1\le i\le i-2$) the tree displayed by $C$ choosing the reticulation arc from leaf $x_i$, and so containing the cherry $\widehat{x_ix_n}$, and $T_{n-1}$ the tree that chooses no reticulation arcs (the underlying tree $T_{n-1}=T_N$).

For each leaf $x_i$ ($1\le i\le n-2$), let the distance from the parent tree vertex to the start of the reticulation arc be $\ell_i$. Label each internal arc from left to right by $m_2,\dots,m_{n-2}$, so that $m_i$ is to the right of $x_i$.
\end{Definition}

We further note here that if $C$ is a caterpillar network, then for any two weighted display trees $T_1^{w_1},T_2^{w_2}$, if $T_1$ and $T_2$ are isomorphic as unweighted trees, they are isomorphic as weighted trees as well, because any weighted display tree of $N$ is uniquely determined by the lowest HGT arc that is not deleted in its formation. In the following lemma, due to this fact we will denote the sum of probabilities assigned across all of the isomorphic weighted copies of $T_i$ by $\beta_{\sum}(T_i)$, noting that this will be a function mapping $\mathcal{T}$ to $(0,1)$.

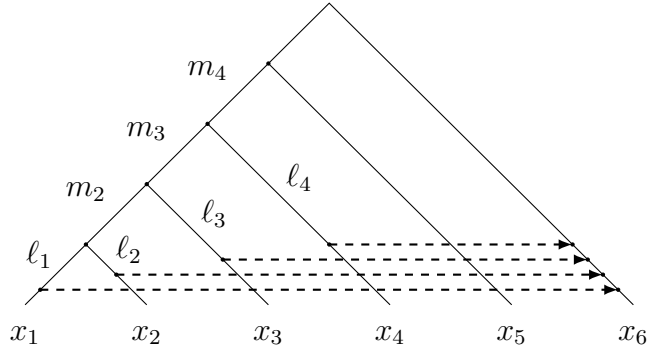
\begin{figure}[H]
\centering
\begin{tikzpicture}[scale=0.4]
\draw (0,0) --(10,10)--(20,0);
\draw (2,2) -- node[above right=-1mm,pos=.4](){$\ell_2$} (4,0);
\draw (4,4) -- node[above right=-1mm,pos=.4](){$\ell_3$} (8,0);
\draw (6,6) -- node[above right=-1mm,pos=.4](){$\ell_4$} (12,0);
\draw (8,8) --(16,0);

\draw [>=latex, ->, dashed, thick] (0.5,0.5) --(19.5,0.5);
\draw [>=latex, ->, dashed, thick] (3,1) --(19,1);
\draw [>=latex, ->, dashed, thick] (6.5,1.5) --(18.5,1.5);
\draw [>=latex, ->, dashed, thick] (10,2) --(18,2);

\draw[fill] (2,2) circle [radius=1.5pt];
\draw[fill] (4,4) circle [radius=1.5pt];
\draw[fill] (6,6) circle [radius=1.5pt];
\draw[fill] (8,8) circle [radius=1.5pt];
\draw[fill] (0.5,0.5) circle [radius=1.5pt];
\draw[fill] (19.5,0.5) circle [radius=1.5pt];
\draw[fill] (6.5,1.5) circle [radius=1.5pt];
\draw[fill] (18.5,1.5) circle [radius=1.5pt];
\draw[fill] (3,1) circle [radius=1.5pt];
\draw[fill] (19,1) circle [radius=1.5pt];
\draw[fill] (10,2) circle [radius=1.5pt];
\draw[fill] (18,2) circle [radius=1.5pt];

\node[above left=-0.01cm] at (3,3) {$\mathstrut m_2$};
\node[above left=-0.01cm] at (5,5) {$\mathstrut m_3$};
\node[above left=-0.01cm] at (7,7) {$\mathstrut m_4$};

\node[above left=-4pt] at (1,1) {$\mathstrut \ell_1$};

\node[below] at (0,0) {$\mathstrut x_1$};
\node[below] at (4,0) {$\mathstrut x_2$};
\node[below] at (8,0) {$\mathstrut x_3$};
\node[below] at (12,0) {$\mathstrut x_4$};
\node[below] at (16,0) {$\mathstrut x_5$};
\node[below] at (20,0) {$\mathstrut x_6$};
\end{tikzpicture}
\caption{A caterpillar network on six leaves,  labelled as required, except for reticulation probabilities.}
\label{f:caterpillar}
\end{figure}

In order to prove that these caterpillar networks are tree-metrizable, we will have to make use of two technical lemmas. The first is Lemma \ref{InequalityLemma}, which we recall is used to reduce the problem of determining whether distances on a network obey the four-point condition to a condition on the sums of the length of the internal arcs of each quartet shape in its display trees.

The second is the next lemma, Lemma \ref{CatDistances}, in which we calculate the sums of the lengths of the internal arcs of each quartet shape in the display trees of a caterpillar network. To this end, we will denote the sum of the length of the internal arcs of each quartet shape $x_ax_b|x_cx_d$ in the display trees of the caterpillar network $C$ by $int(C,x_ax_b|x_cx_d)$. That is, if we denote the set of weighted display trees of $C$ that display $x_ax_b|x_cx_d$ by $\mathcal{T}|_{\{x_ax_b|x_cx_d\}}$, and the length of the internal arc of $T_i|_{\{x_a,x_b,x_c,x_d\}}$ by $p_j$,

\[int(C,x_ax_b|x_cx_d) = \sum_{T_i^w \in \mathcal{T}^w|_{\{x_ax_b|x_cx_n\}}} \beta_{\sum}(T_i) p_i.\]

We are now ready to state and prove the lemma.

\begin{Lemma}
\label{CatDistances}
Let $C$ be a caterpillar network on $X=\{x_1,\dots,x_n\}$, with edge-lengths $\ell_i$ and $m_i$ as shown in Figure~\ref{f:caterpillar}.  Let $\beta_{\sum}: \T \ra (0,1)$ be the function that maps each tree $T_i$ $\in \T_N$ to the total probability assigned to $T_i$ (see Equation~\eqref{eq:beta.def}). For some quartet $q=\{x_a,x_b,x_c,x_n\}$ where $a < b < c < n$, denote the length of the internal arc of $T_i|_{\{x_a,x_b,x_c,x_d\}}$ by $p_i$.
\begin{align*}
int(C,x_ax_b|x_cx_n) & = \beta_{\sum}(T_c)\ell_c +\sum_{s=b+1}^c \left( \sum_{t=s}^{n-1} \beta_{\sum}(T_t) m_{t-1} \right) \\
int(C,x_ax_c|x_bx_n) & = \beta_{\sum}(T_b) \ell_b \\
int(C,x_ax_n|x_bx_c) & = \beta_{\sum}(T_a)\ell_a +\sum_{s=a+1}^b \left( \sum_{t=1}^s \beta_{\sum}(T_t) m_t \right).
\end{align*}
\end{Lemma}

\begin{proof}

We note that $T_i$ will display the quartet $x_ax_b|x_cx_n$ if $i>b$, $x_ax_c|x_bx_n$ if $i=b$, and $x_ax_n|x_bx_c$ if $i < b$, recalling that for a caterpillar network $C$, $T_i$ is a specific tree in $\T_C$, and contains the cherry $\widehat{x_ix_n}$.  

We now consider the internal arcs of each $T_j$. 

If $j < b$, $T_j$ will display $x_ax_n|x_bx_c$ and the internal arc will extend from the least common ancestor of $x_a$ and $x_n$, as depicted in Figure \ref{f:jlessb} (as a dotted line), 
to the parent vertex of $x_b$, so will have a length of $\ell_a + \sum_{k=2}^{b-1} m_k$ in the case that $j=a$, and $\sum_{k=max(a,j)}^{b-1} m_k$ otherwise.

\begin{figure}[ht]
\centering
\begin{tikzpicture}[xscale=.18,yscale=0.27]
\draw[dashed] (0,0) --(2,2);
\draw (8,8) --(10,10)--(20,0);
\draw[thick,dotted] (4,4)--(8,8);
\draw (4,4) --(8,0);
\draw (2,2) --(4,0);
\draw (2,2)--(4,4);
\draw (8,8) --(16,0);
\draw[dashed] (6,6) --(12,0);

\draw[fill] (6,2) circle [radius=1.5pt];
\draw[fill] (4,4) circle [radius=1.5pt];
\draw[fill] (6,6) circle [radius=1.5pt];
\draw[fill] (8,8) circle [radius=1.5pt];

\node[below] at (10,-2) {$\mathstrut T_1$};

\node[below] at (0,0) {$\mathstrut x_1$};
\node[below] at (4,0) {$\mathstrut x_6$};
\node[below] at (8,0) {$\mathstrut x_2$};
\node[below] at (12,0) {$\mathstrut x_3$};
\node[below] at (16,0) {$\mathstrut x_4$};
\node[below] at (20,0) {$\mathstrut x_5$};
\end{tikzpicture}
\quad
\begin{tikzpicture}[xscale=.18,yscale=0.27]
\draw[dashed] (0,0) --(4,4);
\draw (8,8) --(10,10)--(20,0);
\draw[thick,dotted] (6,2)--(4,4)--(8,8);
\draw (6,2) --(4,0);
\draw (6,2) --(8,0);
\draw[dashed] (6,6) --(12,0);
\draw (8,8) --(16,0);

\draw[fill] (2,2) circle [radius=1.5pt];
\draw[fill] (6,6) circle [radius=1.5pt];
\draw[fill] (8,8) circle [radius=1.5pt];
\draw[fill] (6,2) circle [radius=1.5pt];

\node[below] at (10,-2) {$\mathstrut T_2$};

\node[below] at (0,0) {$\mathstrut x_1$};
\node[below] at (4,0) {$\mathstrut x_2$};
\node[below] at (8,0) {$\mathstrut x_6$};
\node[below] at (12,0) {$\mathstrut x_3$};
\node[below] at (16,0) {$\mathstrut x_4$};
\node[below] at (20,0) {$\mathstrut x_5$};
\end{tikzpicture}
\quad
\begin{tikzpicture}[xscale=.18,yscale=0.27]
\draw[dashed] (0,0) --(2,2);
\draw (2,2)--(6,6) (8,8)--(10,10)--(20,0);
\draw[thick,dotted] (6,6)--(8,8);
\draw (2,2) --(4,0);
\draw (6,6) --(12,0);
\draw[dashed] (10,2) --(8,0);
\draw (8,8) --(16,0);

\draw[fill] (2,2) circle [radius=1.5pt];
\draw[fill] (6,6) circle [radius=1.5pt];
\draw[fill] (8,8) circle [radius=1.5pt];
\draw[fill] (10,2) circle [radius=1.5pt];

\node[below] at (10,-2) {$\mathstrut T_3$};

\node[below] at (0,0) {$\mathstrut x_1$};
\node[below] at (4,0) {$\mathstrut x_2$};
\node[below] at (8,0) {$\mathstrut x_3$};
\node[below] at (12,0) {$\mathstrut x_6$};
\node[below] at (16,0) {$\mathstrut x_4$};
\node[below] at (20,0) {$\mathstrut x_5$};
\end{tikzpicture}
\caption{The display trees $T_1,T_2,T_3$ of the six-leaf caterpillar network in Figure \ref{f:caterpillar}. Here we have taken $q=\{2,4,5,6\}$, indicated $T_i|_q$ with filled lines, the internal arc with dotted lines and edges not included in $T_i|_q$ with dashed lines. Note in particular that all three display $x_2x_6|x_4x_5$. In $T_1$, $j=1<a=2$, in $T_2$, $j=2=a$, and in $T_3$, $j=3>a=2$.
}
\label{f:jlessb}
\end{figure}
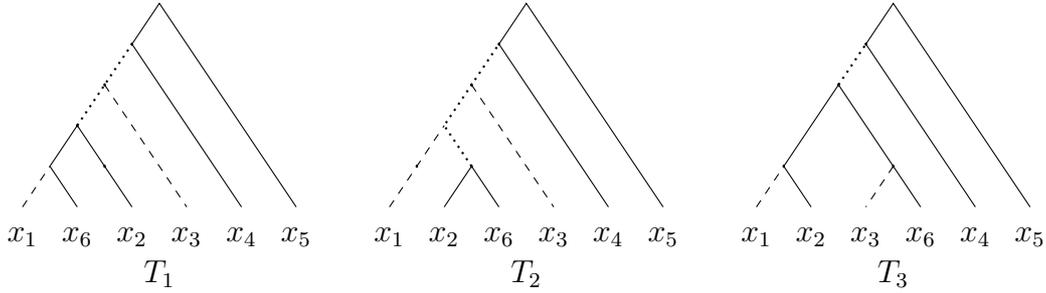

If $j = b$, then $T_j$ will display $x_bx_n|x_ax_c$, and will have an internal arc length of $\ell_b$ (the distance from the least common ancestor of $x_b$ and $x_n$ to its parent vertex).

If $j > b$, then $T_j$ will display $x_ax_b|x_cx_n$ and will have an internal arc extending from the parent vertex of $x_b$ to the parent vertex of whichever of is further left out of $x_c$ and $x_n$, so will have length of $\ell_c + \sum_{k=b}^{c-1} m_k$ in the case that $j=c$ and $\sum_{k=b}^{min(c-1,j-1)} m_k$ otherwise.

If we denote the sum of the contributions from each $T_j$ that displays $x_ax_b|x_cx_n$ by $d_{x_ax_b|x_cx_n}$, it follows that $d_{x_ax_b|x_cx_n}$ will be the sum of the contributions from $T_1$ up to $T_{b-1}$, $d_{x_ax_c|x_bx_n}$ will be the contribution from $T_b$, and $d_{x_ax_n|x_bx_c}$ will be the sum of the contributions from $T_{b+1}$ to $T_{n-1}$. Hence
\begin{align*}
int(C,x_ax_b|x_cx_n)& = \beta_{\sum}(T_c)\ell_c +\sum_{s=b+1}^c \left( \sum_{t=s}^{n-1} \beta_{\sum}(T_t) m_{t-1} \right) \\
int(C,x_ax_c|x_bx_n)& = \beta_{\sum}(T_b) \ell_b \\
int(C,x_ax_n|x_bx_c) & = \beta_{\sum}(T_a)\ell_a +\sum_{s=a+1}^b \left( \sum_{t=1}^s \beta_{\sum}(T_t) m_t \right).
\end{align*}
as required.
\end{proof}

Now that we have proven the technical lemma, the main theorem follows.

\begin{Theorem}
\label{Caterpillar}
Let $C$ be a caterpillar network. Then $C$ is tree-metrizable on every tree it displays.
\end{Theorem}

\begin{proof}
For each leaf $x_i$ $(i<n-1)$, let the distance from the parent tree vertex to the start of the reticulation arc be $\ell_i$, and label each internal arc from left to right by $a_2,\dots,a_{n-2}$, so that $m_i$ is to the right of $x_i$, as in Figure~\ref{f:caterpillar}. Let $\beta_{\sum}: \T \ra (0,1)$ be the function that maps each tree $T_i$ to the total probability assigned to $T_i$. As all weighted display trees of $C$ that are isomorphic as unweighted trees are also isomorphic as weighted trees, this function can be defined unambiguously.

Recall that there are $n-1$ isomorphism classes of display trees of $C$, $\T_C=\{T_1,\dots,T_{n-1}\}$, where $T_i$ is the unique tree for which $x_i$ has $x_n$ as its closest neighbour for each $i \in \{1,\dots,n-2\}$, plus the underlying tree $T_{n-1}$.  We will show that $C$ is $T_i$-metrizable for each $i$.

Consider a quartet $q=\{ x_a,x_b,x_c,x_d \}$, supposing without loss of generality that $a<b<c<d$. We will now find the internal arc weights of the weighted display trees, with a view to invoking Lemma \ref{InequalityLemma}.

We first suppose that $n \ne a,b,c,d$. Any weighted display tree restricted to $q$ is isomorphic to $T_C^{w(T_C)}\vert_q$ , where $w(T_C)$ is the induced weighting on the underlying tree, and furthermore $T_i|_q$ is isomorphic as an unweighted tree to $T_C|_q$, because $x_n$ is not involved in $q$ and so none of the reticulation arcs are involved in either $T_C|_q$ or $T_i|_q$.  Therefore the four-point condition will be obeyed for all such $q$, and $N|_q$ is tree-metrizable on $T_i|_q = T_C|_q$.

Define the functions
\begin{align*}
\gamma^+(j) & = \frac{\beta_{\sum}(T_{j-1}) \ell_{j-1} - \sum_{k=j+1}^{n-1} \beta_{\sum}(T_k) m_{k-1}}{\beta_{\sum}(T_j)}\\
\gamma^-(j) & = \frac{\beta_{\sum}(T_{j+1}) \ell_{j+1} - \sum_{k=1}^{j-1} \beta_{\sum}(T_k) m_k}{\beta_{\sum}(T_j)}, \\
\end{align*}
noting that the function $\gamma^+$ is undefined for $j=n-2$, and similarly $\gamma^-$ is undefined for $j=1$. Note that as $n>3$ for a caterpillar network every $j$ has at least one output between the two functions.

Fix all $a_j$ to be some arbitrary non-zero lengths. Let $\ell_1,\dots,\ell_{i-1},\ell_{i+1},\dots,\ell_{n-2}$ be positive solutions to the following system of linear equations. 
\begin{equation*}
\ell_j =\begin{cases}
\ell_{n-2}, & \text{if $j = n-3$} \\
\gamma^+(j) , & \text{if $i +1 < j < n-3$},\\
\gamma^-(j) , & \text{if $2<j< i-1$},\\
\ell_1, & \text{if $j=2$},
\end{cases}
\end{equation*}
Note that either $\ell_{i+1}$  or $\ell_{i-1}$ will not exist if $i=n-2$ or $i=1$ respectively, but if both exist, we further require that 
\[\beta_{\sum}(T_{i+1}) \ell_{i+1}+\sum_{k=i+1}^{n-1}\beta_{\sum}(T_k)m_k = \beta_{\sum}(T_{i-1})\ell_{i-1}+\sum_{k=1}^{i-1}\beta_{\sum}(T_k) m_k.\]
It is a simple exercise in linear algebra that there exist strictly positive values of $\ell_j$ for all $j \ne i$ that satisfy these equations.

Finally, set $\ell_i$ to be any value larger than $\max\{\gamma^+(i), \gamma^-(i)\}$, where if either $\gamma^+(i)$ or $\gamma^-(i)$ are undefined we just require $\ell_i$ to be larger than the existing expression. We claim that this causes the quartet $x_a,x_b,x_c,x_n$ to satisfy the conditions of Lemma \ref{InequalityLemma}.

In particular, we claim that $C|_q$ displays $x_ax_b|x_cx_n$ if $i>b$, $x_ax_c|x_bx_n$ if $i=b$, and $x_ax_n|x_bx_c$ if $i < b$. To show this, we must show that the appropriate sum of internal arcs is larger than the other two, which must be equal as per Lemma \ref{InequalityLemma}.

First suppose $i>b$. Then to satisfy Lemma \ref{InequalityLemma} we must show
\[\sum_{T_j \cong x_ax_b|x_cx_n} p_j \beta_{\sum}(T_j) > \sum_{T_j \cong x_ax_c|x_bx_n} p_j \beta_{\sum}(T_j) = \sum_{T_j \cong x_ax_n|x_bx_c} p_j \beta_{\sum}(T_j).\]
Therefore we require
\[d_{x_ax_b|x_cx_n} > d_{x_ax_c|x_bx_n} = d_{x_ax_n|x_bx_c}.\]
We first check the equality condition. From Lemma \ref{CatDistances} we know 
\begin{align*}
int(C,x_ax_c|x_bx_n) & = \beta_{\sum}(T_b) \ell_b \\
& = \beta_{\sum}(T_{b-1}) \ell_{b-1} + \sum_{k=1}^b \beta_{\sum}(T_k) m_k \\ 
& = \beta_{\sum}(T_{b-2}) \ell_{b-2} + \sum_{k=1}^{b-1} \beta_{\sum}(T_k) m_k + \sum_{k=1}^b \beta_{\sum}(T_k) m_k \\ 
& = \dots \\
& = \beta_{\sum}(T_a)\ell_a +\sum_{s=a+1}^b \left( \sum_{t=1}^s \beta_{\sum}(T_t) m_t \right) \\ 
& = int(C,x_ax_n|x_bx_c)
\end{align*}
We then check the inequality condition, which must be checked in two parts - for $i \ge c$ or $i < c$. First observe that where both exist, $\beta_{\sum}(T_j) \gamma^+(j) > \beta_{\sum}(T_{j+1}) \ell_{j+1}$, and similarly that $\beta_{\sum}(T_j) \gamma^-(j) > \beta_{\sum}(T_{j-1}) \ell_{j-1}$. Now, if $i \ge c$, from Lemma \ref{CatDistances} we know
\begin{align*}
int(C,x_ax_b|x_cx_n) & = \beta_{\sum}(T_c) \ell_c +\sum_{s=b+1}^c \left( \sum_{t=s}^{n-1} \beta_{\sum}(T_t) m_{t-1} \right) \\
& > \beta_{\sum}(T_c) \ell_c \\
& > \beta_{\sum}(T_{c-1}) \ell_{c-1} \\
& > \dots \\
& > \beta_{\sum}(T_b) \ell_b \\
& = int(C,x_ax_c|x_bx_n). \\
\end{align*}
Otherwise, if $i < c$
\begin{align*}
int(C,x_ax_b|x_cx_n) & = \beta_{\sum}(T_c)\ell_c +\sum_{s=b+1}^c \left( \sum_{t=s}^{n-1} \beta_{\sum}(T_t) m_{t-1} \right) \\
& = \beta_{\sum}(T_{c-1}) \ell_{c-1} +\sum_{s=b+1}^{c-1} \left( \sum_{t=s}^{n-1} \beta_{\sum}(T_t) m_{t-1} \right)\\
& = \dots \\
& = \beta_{\sum}(T_{i+1})\ell_{i+1} +\sum_{s=b+1}^{i+1} \left( \sum_{t=s}^{n-1} \beta_{\sum}(T_t) m_{t-1} \right)\\
& = \beta_{\sum}(T_{i-1})\ell_{i-1}+\sum_{k=1}^{i-1}\beta_{\sum}(T_k)m_k \\
& > \beta_{\sum}(T_{i-1}) \ell_{i-1} \\
& > \dots \\
& > \beta_{\sum}(T_b) \ell_b \\
& = int(C,x_ax_c|x_bx_n). \\
\end{align*}
This proves the case where $i >b$. The cases where $i=b$ and $i<b$ are proved similarly. 
\end{proof}

\begin{Corollary}\label{c:height.h}
Let $T$ be a tree of height $h>2$. Then there exists a $T$-metrizable network with underlying tree $T$ and $h-1$ non-trivial reticulations.
\end{Corollary}

\begin{proof}
Let $T^{cat}$ be the caterpillar tree on $h+1$ leaves, with leaves $Y=\{y_1,\dots,y_{h+1} \}$. As $T$ is of height $h$, there exists a graph embedding $\delta$ from $T^{cat}$ into $T$, for instance by mapping the `backbone' of the caterpillar to a path of length $h$ in $T$.

Let $T_i$ be the subtree of $T$ induced by $\delta(y_i)$ for each $i \in \{1,\dots,h+1\}$. Now, by Theorem \ref{Caterpillar}, there exists a network $N$ on $Y$ with caterpillar underlying tree that is $T^{cat}$-metrizable (specifically the network with base tree $T^{cat}$). Note further that $N$ has $h-1$ non-trivial arcs. By repeated application of Theorem \ref{t:LeafReplacement}, we can graft each of $T_1,\dots,T_{h+1}$ to each of $y_1,\dots,y_{h+1}$ respectively in $N$, and the resulting network will be $T$-metrizable.
\end{proof}

With Corollary \ref{c:height.h} we have now shown that for every tree $T$ of height at least $3$ there exists a non-trivial $T$-metrizable HGT network. Additionally, with Theorem \ref{Caterpillar} we have found an infinite class of $T$-metrizable phylogenetic networks where $T$ is not the underlying tree. Together, these results extend the surprising result of Theorem \ref{QuartetExample} in two different directions. Finally, Theorem \ref{t:LeafReplacement} shows us that we can form a tree-metrizable network very easily by grafting trees onto leaves of a network. However, we have not yet considered the opposite case - grafting networks onto the leaves of a tree. We will address this in the next section.

\section{Leaf-Grafts with Network Scions}\label{s:tree.stock}

The question of whether we can form a tree-metrizable network by leaf-grafting a network onto a tree is more complicated. For instance, consider the network $N_1$ shown in Figure \ref{NotTMNetwork}. It is formed by leaf-grafting $N$, a slight modification of the network from \cite{francis2015tree}, onto a $2$-leaf binary tree. The network $N$ is tree-metrizable by a combination of Theorem \ref{QuartetExample} and Lemma \ref{OmitAdjacent}. However, despite the fact that $N_1$ is formed by grafting a tree-metrizable network into a tree, $N_1$ is not itself tree-metrizable - restricting its display trees to $\{x_1,x_3,x_4,x_5\}$ only gives $x_1x_5|x_3x_4$ and $x_1x_4|x_3x_5$, which by Lemma \ref{NotTwo} implies that $N_1$ is not tree-metrizable.

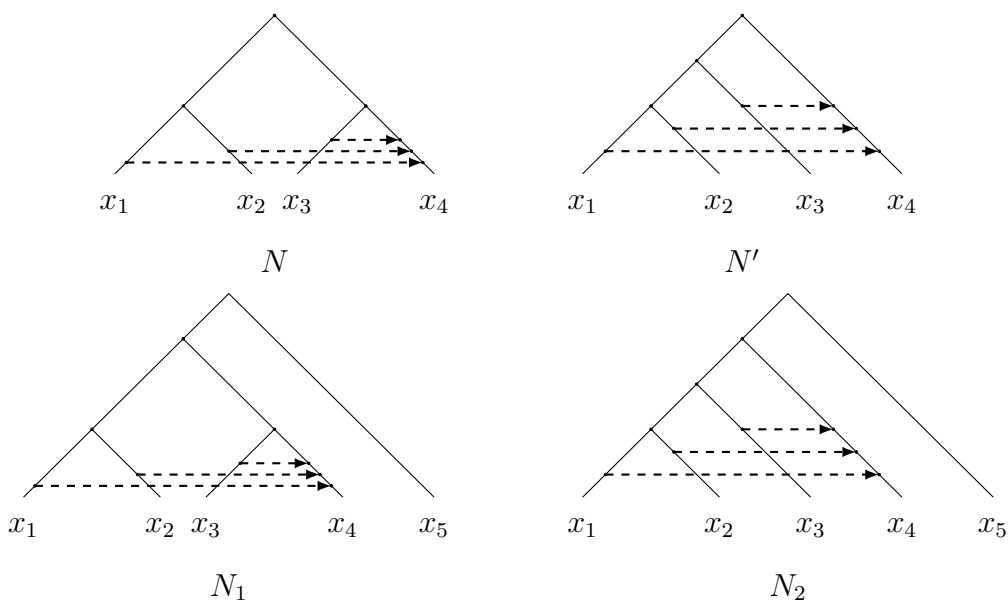
\begin{figure}[ht]
\centering
\begin{tikzpicture}[scale=0.3]
\draw (0,0) --(7,7)--(14,0);
\draw (3,3) --(6,0);
\draw (11,3) --(8,0);

\draw [>=latex, ->, dashed, thick] (5,1) --(13,1);
\draw [>=latex, ->, dashed, thick] (0.5,0.5) --(13.5,0.5);
\draw [>=latex, ->, dashed, thick] (9.5,1.5) --(12.5,1.5);

\draw[fill] (3,3) circle [radius=1.5pt];
\draw[fill] (11,3) circle [radius=1.5pt];
\draw[fill] (7,7) circle [radius=1.5pt];
\draw[fill] (5,1) circle [radius=1.5pt];
\draw[fill] (13,1) circle [radius=1.5pt];
\draw[fill] (0.5,0.5) circle [radius=1.5pt];
\draw[fill] (13.5,0.5) circle [radius=1.5pt];
\draw[fill] (9.5,1.5) circle [radius=1.5pt];
\draw[fill] (12.5,1.5) circle [radius=1.5pt];

\node[below] at (0,0) {$\mathstrut x_1$};
\node[below] at (6,0) {$\mathstrut x_2$};
\node[below] at (8,0) {$\mathstrut x_3$};
\node[below] at (14,0) {$\mathstrut x_4$};

\node at (7,-4) {$\mathstrut N$};
\end{tikzpicture}
\hspace{1cm}
\begin{tikzpicture}[scale=0.3]
\draw (0,0) --(7,7)--(14,0);
\draw (3,3) --(6,0);
\draw (5,5) --(10,0);

\draw [>=latex, ->, dashed, thick] (7,3) --(11,3);
\draw [>=latex, ->, dashed, thick] (4,2) --(12,2);
\draw [>=latex, ->, dashed, thick] (1,1) --(13,1);

\draw[fill] (3,3) circle [radius=1.5pt];
\draw[fill] (1,1) circle [radius=1.5pt];
\draw[fill] (13,1) circle [radius=1.5pt];
\draw[fill] (4,2) circle [radius=1.5pt];
\draw[fill] (12,2) circle [radius=1.5pt];
\draw[fill] (5,5) circle [radius=1.5pt];
\draw[fill] (7,7) circle [radius=1.5pt];
\draw[fill] (7,3) circle [radius=1.5pt];
\draw[fill] (11,3) circle [radius=1.5pt];

\node[below] at (0,0) {$\mathstrut x_1$};
\node[below] at (6,0) {$\mathstrut x_2$};
\node[below] at (10,0) {$\mathstrut x_3$};
\node[below] at (14,0) {$\mathstrut x_4$};

\node at (7,-4) {$\mathstrut N'$};
\end{tikzpicture}
\begin{tikzpicture}[scale=0.3]
\draw (0,0) --(7,7)--(14,0);
\draw (7,7)--(9,9)--(18,0);
\draw (3,3) --(6,0);
\draw (11,3) --(8,0);

\draw [>=latex, ->, dashed, thick] (5,1) --(13,1);
\draw [>=latex, ->, dashed, thick] (0.5,0.5) --(13.5,0.5);
\draw [>=latex, ->, dashed, thick] (9.5,1.5) --(12.5,1.5);

\draw[fill] (3,3) circle [radius=1.5pt];
\draw[fill] (11,3) circle [radius=1.5pt];
\draw[fill] (7,7) circle [radius=1.5pt];
\draw[fill] (5,1) circle [radius=1.5pt];
\draw[fill] (13,1) circle [radius=1.5pt];
\draw[fill] (0.5,0.5) circle [radius=1.5pt];
\draw[fill] (13.5,0.5) circle [radius=1.5pt];
\draw[fill] (9.5,1.5) circle [radius=1.5pt];
\draw[fill] (12.5,1.5) circle [radius=1.5pt];

\node[below] at (0,0) {$\mathstrut x_1$};
\node[below] at (6,0) {$\mathstrut x_2$};
\node[below] at (8,0) {$\mathstrut x_3$};
\node[below] at (14,0) {$\mathstrut x_4$};
\node[below] at (18,0) {$\mathstrut x_5$};

\node at (9,-4) {$\mathstrut N_1$};
\end{tikzpicture}
\hspace{1cm}
\begin{tikzpicture}[scale=0.3]
\draw (0,0) --(7,7)--(14,0);
\draw (7,7)--(9,9)--(18,0);
\draw (3,3) --(6,0);
\draw (5,5) --(10,0);

\draw [>=latex, ->, dashed, thick] (7,3) --(11,3);
\draw [>=latex, ->, dashed, thick] (4,2) --(12,2);
\draw [>=latex, ->, dashed, thick] (1,1) --(13,1);

\draw[fill] (3,3) circle [radius=1.5pt];
\draw[fill] (1,1) circle [radius=1.5pt];
\draw[fill] (13,1) circle [radius=1.5pt];
\draw[fill] (4,2) circle [radius=1.5pt];
\draw[fill] (12,2) circle [radius=1.5pt];
\draw[fill] (5,5) circle [radius=1.5pt];
\draw[fill] (7,7) circle [radius=1.5pt];
\draw[fill] (7,3) circle [radius=1.5pt];
\draw[fill] (11,3) circle [radius=1.5pt];

\node[below] at (0,0) {$\mathstrut x_1$};
\node[below] at (6,0) {$\mathstrut x_2$};
\node[below] at (10,0) {$\mathstrut x_3$};
\node[below] at (14,0) {$\mathstrut x_4$};
\node[below] at (18,0) {$\mathstrut x_5$};

\node at (9,-4) {$\mathstrut N_2$};
\end{tikzpicture}
\caption{Two examples of a networks formed by leaf-grafting a tree-metrizable network onto a tree. The resulting network $N_1$ is not tree-metrizable, but the network $N_2$ is tree-metrizable.}
\label{NotTMNetwork}
\end{figure}

However, if we take the network $N$ and just move the root to leaf $x_4$ to form $N'$, and \textit{then} graft it, we obtain the network in Figure \ref{NotTMNetwork} (ii), which is tree-metrizable by the observation that it is an example from Theorem \ref{Caterpillar} with a relocation of the root for which the pairwise distances for any weighting do not change (for details, see Theorem \ref{RootRel}). 

The remainder of the paper has two aims. Firstly, we classify all possible level-$2$ networks $N$ for which the leaf-graft of $N$ onto a tree $T$ will be tree-metrizable. Subsequently, we will define a class of networks for which the leaf-graft of a network onto a tree $T$ will always produce a tree-metrizable network.

\begin{Definition}
Let $N$ be an HGT network, with underlying tree $T_N$. Suppose that some pair of edges $a_1,a_2$ of $T_N$ are subdivided and an HGT arc placed between them in either direction. Then we say $a_1,a_2$ are \textit{HGT-connected} and denote it $a_1-a_2$. 
\end{Definition}

As we will often be considering the incoming edges of the leaves, call such an edge a \textit{leaf arc}, and if the leaf is denoted $x_i$, denote the leaf arc of $x_i$ by $e_i$.

The following theorem greatly reduces the possibilities for tree-metrizable quartets and will serve as a useful tool for the following theorems.

\begin{Theorem}
\label{3IsoClass}
Let $N$ be a network on $q=\{x_1,x_2,x_3,x_4\}$ with an underlying tree of the form $x_1x_2|x_3x_4$. Then $N$ displays three isomorphism classes of trees if and only if there exists one leaf in $q$ that is HGT-connected to both of the other non-adjacent leaves (e.g. $x_1-x_3$ and $x_1-x_4$).
\end{Theorem}

\begin{proof}
Let $e_i$ denote the leaf arc of $x_i$ in \textit{the underlying tree}. First suppose there exists one leaf $x_i$ in $q$ so that $e_i$ is HGT-connected to both of the other non-adjacent leaf arcs. Without loss of generality suppose $N$ has two reticulation arcs $a,b$ that connect the leaf $e_1$ to $e_3$ and $e_1$ to $e_4$ respectively. Observe that the display tree obtained by deleting all HGT arcs except for $a$ has the cherry $\widehat{x_1x_3}$, the display tree obtained by deleting all HGT arcs except $b$ has the cherry $\widehat{x_1x_4}$. Thus these display trees are non-isomorphic as unrooted trees, and neither is isomorphic to the underlying tree (which is always displayed as we can just delete all HGT arcs). Hence $N$ displays three isomorphism classes.

Now suppose that $N$ does not have one leaf in $q$ that is HGT-connected to both of the other non-adjacent leaves. We shall consider each possible scenario.

Firstly, suppose there are at least two arcs that connect $\{e_1,e_2\}$ to $\{e_3,e_4\}$, but that no leaf arc is HGT-connected to both of the leaf arcs in the other set. That is, that $e_1-e_3$ and $e_2-e_4$ or $e_1-e_4$ and $e_2-e_3$.  These cases are identical up to relabelling, so it suffices to consider $e_1-e_3$, $e_2-e_4$. It is then easy to observe that in this case we obtain only two display trees namely $x_1x_2|x_3x_4$ and $x_1x_3|x_2x_4$).

Now suppose that there is exactly one leaf arc in $\{e_1,e_2\}$ HGT-connected to a leaf arc in $\{e_3,e_4\}$. Without loss of generality, assume it is $e_1-e_3$. Observe that $N$ therefore cannot display $e_1e_4|e_2e_3$, as this requires $e_2-e_3$ or $e_1-e_4$, and so $N$ does not display all three isomorphism classes.

Finally, suppose that no leaf arc in $\{e_1,e_2\}$ is HGT-connected to a leaf arc in $\{e_3,e_4\}$. Then all reticulation arcs are between adjacent arcs, so by Lemma \ref{AdjacentLemma}, $N$ only displays the underlying tree.
\end{proof}

In order to examine our graftings of networks into trees, we will need to isolate them from the rest of the tree. In order to do this we shall use the concept of a biconnected component.

\begin{Definition}\label{d:biconnected}
A \textit{biconnected component} of an HGT network $N$ is a maximal HGT-connected subgraph $B$ of $N$ for which the removal of any arc of $B$ is a HGT-connected graph.
\end{Definition}

Observe that if $V(B)$ is the collection of vertices contained in a biconnected component, then $LSA(V(B))$ is always contained in $B$. However, a biconnected component is never going to be a binary network: it cannot contain any leaves as the leaf vertices are degree-$1$, and so the removal of a leaf arc will result in a disconnected graph. We therefore must find the smallest sub-network of our network that contains $B$, including all necessary information. For this purpose we will use the concept of induced networks \cite[p.~143]{Huson2010}, but with a small modification. 

\begin{Definition}\label{d:min.supp.net}
Let $B$ be a biconnected component of an HGT network $N$ on $X$. Let $\{v_1,\dots,v_k\}$ be the set of vertices in $B$ that have smaller outdegree in $B$ than in $N$. For each $v_i$ add a vertex $w_i$ and the edge $(v_i,w_i)$. Finally, label all the resulting leaves by a unique descendent of $v_i$, so that all leaves have a distinct label. The resulting phylogenetic network is referred to as a \textit{minimal support network of $B$ in $N$}, denoted $N(B)$, and is unique up to rearrangement of the leaf labels.
\end{Definition}

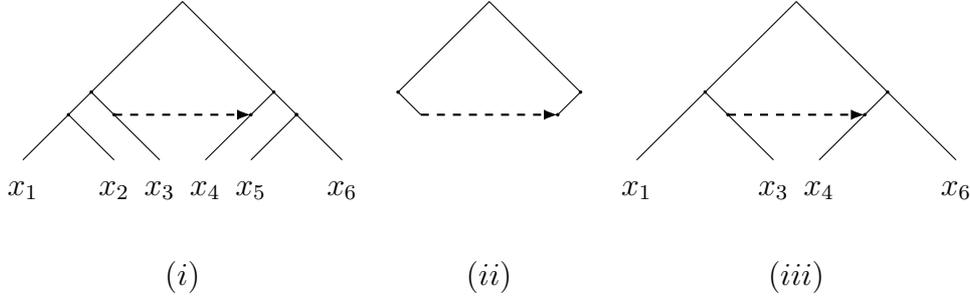
\begin{figure}[H]
\centering
\begin{tikzpicture}[scale=0.3]
\draw (0,0) --(7,7)--(14,0);
\draw (2,2) --(4,0);
\draw (3,3) --(6,0);
\draw (11,3) --(8,0);
\draw (12,2) --(10,0);

\draw [>=latex, ->, dashed, thick] (4,2) --(10,2);

\draw[fill] (2,2) circle [radius=1.5pt];
\draw[fill] (12,2) circle [radius=1.5pt];
\draw[fill] (3,3) circle [radius=1.5pt];
\draw[fill] (11,3) circle [radius=1.5pt];
\draw[fill] (12,2) circle [radius=1.5pt];
\draw[fill] (4,2) circle [radius=1.5pt];
\draw[fill] (10,2) circle [radius=1.5pt];

\node[below] at (0,0) {$\mathstrut x_1$};
\node[below] at (4,0) {$\mathstrut x_2$};
\node[below] at (6,0) {$\mathstrut x_3$};
\node[below] at (8,0) {$\mathstrut x_4$};
\node[below] at (10,0) {$\mathstrut x_5$};
\node[below] at (14,0) {$\mathstrut x_6$};

\node[below] at (7,-4) {$\mathstrut (i)$};
\end{tikzpicture}
\hspace{0.1cm}
\begin{tikzpicture}[scale=0.3]
\draw (3,3) --(7,7)--(11,3);
\draw (3,3) --(4,2);
\draw (11,3) --(10,2);

\draw [>=latex, ->, dashed, thick] (4,2) --(10,2);

\draw[fill] (3,3) circle [radius=1.5pt];
\draw[fill] (11,3) circle [radius=1.5pt];
\draw[fill] (10,2) circle [radius=1.5pt];

\node[below] at (7,-4) {$\mathstrut (ii)$};
\end{tikzpicture}
\hspace{0.1cm}
\begin{tikzpicture}[scale=0.3]
\draw (0,0) --(7,7)--(14,0);
\draw (3,3) --(6,0);
\draw (11,3) --(8,0);

\draw [>=latex, ->, dashed, thick] (4,2) --(10,2);

\draw[fill] (3,3) circle [radius=1.5pt];
\draw[fill] (11,3) circle [radius=1.5pt];
\draw[fill] (4,2) circle [radius=1.5pt];
\draw[fill] (10,2) circle [radius=1.5pt];

\node[below] at (0,0) {$\mathstrut x_1$};
\node[below] at (6,0) {$\mathstrut x_3$};
\node[below] at (8,0) {$\mathstrut x_4$};
\node[below] at (14,0) {$\mathstrut x_6$};

\node[below] at (7,-4) {$\mathstrut (iii)$};
\end{tikzpicture}
\caption{(i) an HGT network $N$ with a single non-trivial biconnected component, $B$; (ii) The non-trivial biconnected component $B$ of $N$; (iii) A minimal support network of $B$.}
\label{BCPic}
\end{figure}

We note that minimal support networks are perhaps most easily understood as the rooted equivalent of $B_N$ from \cite{francis2018a}. We also note here that if a minimal support network of a biconnected component in a network $N$ is not tree-metrizable, it easily follows that $N$ is not tree-metrizable.

\begin{Theorem}
Let $N$ be a level-$2$ tree-metrizable network, containing a biconnected component $B$. Then any minimal support network $N(B)$ in $N$ has a caterpillar underlying tree unless either 

\begin{enumerate}
\item $N(B)$ contains the root of $N$, or 
\item all reticulation arcs in $N(B)$ are between adjacent arcs of $T_{N(B)}$.
\end{enumerate}
\end{Theorem}

\begin{proof}
Suppose, seeking a contradiction, that the underlying tree $T_{N(B)}$  of $N(B)$ contains two cherries, since a tree is a caterpillar tree if and only if it has exactly one cherry. Additionally suppose $N(B)$ does not contain the root.  First observe that if $N(B)$ contains exactly one reticulation arc, and it is not between adjacent arcs of $T_{N(B)}$,  $N(B)$ is not tree-metrizable~\cite[Theorem 5]{francis2015tree} and it follows that $N$ is not tree-metrizable. Hence we can assume that there are two reticulation arcs in $N(B)$, and at least one of them is not between adjacent arcs of the underlying tree.

Let $q=x_1x_2|x_3x_4$ be a quartet in $N(B)$ such that $q$ has two cherries of minimal height, that is, there is no leaf $\ell$ that separates $x_1$ from $x_2$, or $x_3$ from $x_4$. Denote the leaf edges of $x_i$ \textit{in the underlying tree of $N(B)$, $T_{N(B)}$}, by $e_i$. Then, by construction, both cherries of $N_q$ contain the source or target of a reticulation arc, since $N(B)$ is a minimal support network and so the tree vertex parent of the underlying tree of $N(B)$ must have at least one reticulation descendant. Without loss of generality, suppose $e_1$ and $e_3$ are the sources or targets of  reticulation arcs. 

We first assume that at least one of the HGT arcs starting or ending on $e_1$ and $e_3$ has the other end on an arc outside of $N|_q$. We will show that this means $N$ is not tree-metrizable, by running through the cases.

Suppose without loss of generality, that $e_1$ is HGT-connected to an arc $A$ outside of $N|_q$, and that $A$ has some descendant $\ell$.

\begin{enumerate}
\item \textit{Either $e_2$ or $e_4$ are HGT-connected to $e_3$:} Suppose first that $e_2$ or $e_4$ is HGT-connected to $e_3$. If $e_2-e_3$, $N(B)$ is not tree-metrizable by considering $N(B)|_{\{x_1,x_2,x_3,\ell \} }$ with Theorem \ref{3IsoClass}. If $e_3-e_4$, then as no other arcs connect to $e_3$ or $e_4$ and $x_3$ forms a cherry with $x_4$, this is a case of an arc connecting siblings immediately below a tree vertex, which can be omitted by Lemma \ref{OmitAdjacent}. This leaves a single arc between non-adjacent arcs, which by Lemma \ref{NotTwo} implies $N$ is not tree-metrizable. We can therefore assume that $e_3$ is not HGT-connected to $e_2$ or $e_4$. 

\item \textit{Neither $e_2$ nor $e_4$ are HGT-connected to $e_3$ and the endpoints of the arcs HGT-connected to $e_1$ and $e_3$ are connected in $N$ by a tree-path:}  First, suppose the arc HGT-connected to $e_1$ is a tree-path descendant of the arc HGT-connected to $e_3$, or vice versa. Then we can select $\ell$ to be a descendant of both and consider the network induced by $\{x_1,x_2,x_3,\ell\}$. It is clear that there exist display trees with $x_1x_2|x_3\ell$ and $x_1x_3|x_2\ell$, but there are none with $x_1 \ell |x_2x_3$, since $x_1$ and $x_2$ either form a cherry or $x_2$ and $\ell$ form a cherry. Thus $N$ cannot be tree-metrizable unless $e_1$ and $e_3$ are HGT-connected to arcs that are not descended from one another, by Lemma~\ref{NotTwo}. 

\item \textit{Neither $e_2$ nor $e_4$ are HGT-connected to $e_3$ and the endpoints of the arcs HGT-connected to $e_1$ and $e_3$ are not connected in $N$ by a tree-path:} Let the arcs HGT-connected to $e_1$ and $e_3$ be denoted $A$ and $B$ respectively. As $A$ and $B$ are not descendants of one another, there exists $\ell$ and $k$ that are each only descended from  $A$ and $B$ respectively. If the underlying tree $T_{N(B)}$ of $N(B)$ restricted to $\{x_1,x_3,k,\ell \}=x_1x_3|k\ell$ or $x_1\ell | x_3 k$, then $N(B)$ is not tree-metrizable by Lemma \ref{3IsoClass}. However, if $T_{N(B)}$ restricted to $\{x_1,x_3,k,\ell \}$ is $x_1k|x_3\ell$, then we can consider the reticulation $a$ between $x_3$ and $\ell$. If $a$ ends on $x_3$, we consider $q_1= \{x_1,x_2,k,\ell\}$, and if $a$ ends on $\ell$, we consider $q_2=\{x_1,x_2,x_3,k\}$. In both cases, the set of $N(B)$'s display trees restricted to $q_i$ has exactly two non-isomorphic elements and so by Lemma \ref{NotTwo}, $N(B)$ --- and thus $N$ --- is not tree-metrizable.
\end{enumerate}

We can therefore assume all reticulation arcs in the biconnected component are between arcs of $N_q$, that is, that $N_B=N_q$.

If there are no arcs between the two cherries then we have the trivial case of all reticulation arcs in $N(B)$ being between adjacent arcs. Hence assume $e_1-e_3$. By Theorem \ref{3IsoClass}, $N_q$ needs either $e_2-e_3 $ or $e_1-e_4$ in order to display all $3$ isomorphism classes and thus be tree-metrizable. 
By symmetry, it suffices to consider the $e_1-e_3, e_2-e_3$ case.

Now, as $N_q$ does not contain the root, there exists a leaf $\ell$ attached to the central arc of $N_q$ (omitting reticulation arcs), as in Figure~\ref{f:leaf.btw.2.cherries}.

\begin{figure}[ht]
\centering
\begin{tikzpicture}[scale=0.4]
\draw (0,0) --(5,5)--(10,0);
\draw (5,5)--(6,6)--(12,0);
\draw (2,2) --(4,0);
\draw (8,2) --(6,0);

\draw[fill] (2,2) circle [radius=1.5pt];
\draw[fill] (8,2) circle [radius=1.5pt];
\draw[fill] (5,5) circle [radius=1.5pt];

\node[below] at (0,0) {$\mathstrut x_1$};
\node[below] at (4,0) {$\mathstrut x_2$};
\node[below] at (6,0) {$\mathstrut x_3$};
\node[below] at (10,0) {$\mathstrut x_4$};
\node[below] at (12,0) {$\mathstrut \ell$};
\end{tikzpicture}
\caption{A leaf $\ell$ attached to the edge between two cherries.}\label{f:leaf.btw.2.cherries}
\end{figure}
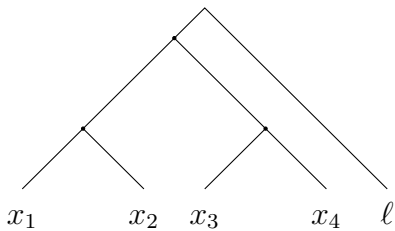

One can quickly observe that the display trees obtained by this configuration display $x_1 \ell |x_2x_4$ and $x_1x_2|x_4 \ell$ but not $x_1x_4|x_2 \ell$, so $N$ is not tree-metrizable by Lemma \ref{NotTwo}.

As all possibilities are exhausted, the theorem follows.
\end{proof}

We will now define a class of networks $\mathcal{C}$ for which all leaf-grafts of $N \in \mathcal{C}$ onto a tree $T$ are tree-metrizable.

\begin{Definition}
Let $N$ be a network obtained by taking a caterpillar network $C$ and adding a reticulation arc from leaf edge $e_{n-1}$ to leaf edge $e_n$, ending above all of the others. Then $N$ is termed an \textit{enhanced caterpillar network}.
\end{Definition}

\begin{Theorem}
\label{RootRel}
Let $N$ be an enhanced caterpillar network. Let $T$ be a tree. Then any network $N'$ formed by leaf-grafting $N$ onto $T$ is tree-metrizable.
\end{Theorem}

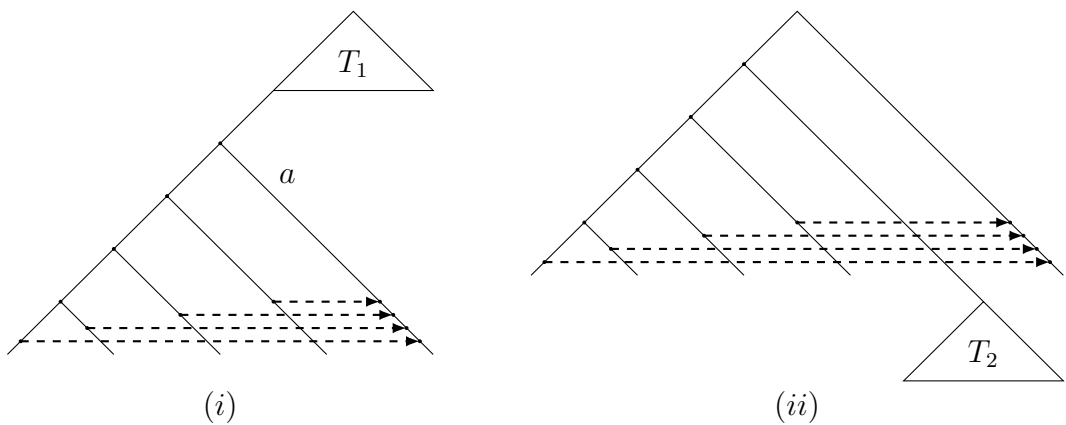
\begin{figure}[ht]
\centering
\begin{tikzpicture}[scale=0.35]
\draw (0,0) --(10,10);
\draw (2,2) --(4,0);
\draw (4,4) --(8,0);
\draw (6,6) --(12,0);
\draw (8,8) --(16,0);
\draw (10,10)--(13,13)--(16,10)--(10,10);

\draw [>=latex, ->, dashed, thick] (0.5,0.5) --(15.5,0.5);
\draw [>=latex, ->, dashed, thick] (3,1) --(15,1);
\draw [>=latex, ->, dashed, thick] (6.5,1.5) --(14.5,1.5);
\draw [>=latex, ->, dashed, thick] (10,2) --(14,2);

\draw[fill] (2,2) circle [radius=1.5pt];
\draw[fill] (4,4) circle [radius=1.5pt];
\draw[fill] (6,6) circle [radius=1.5pt];
\draw[fill] (8,8) circle [radius=1.5pt];
\draw[fill] (0.5,0.5) circle [radius=1.5pt];
\draw[fill] (15.5,0.5) circle [radius=1.5pt];
\draw[fill] (3,1) circle [radius=1.5pt];
\draw[fill] (15,1) circle [radius=1.5pt];
\draw[fill] (6.5,1.5) circle [radius=1.5pt];
\draw[fill] (14.5,1.5) circle [radius=1.5pt];
\draw[fill] (10,2) circle [radius=1.5pt];
\draw[fill] (14,2) circle [radius=1.5pt];

\node at (13,11) {$\mathstrut T_1$};
\node[above right=-0.1cm] at (10,6) {$\mathstrut a$};

\node at (8,-2) {$\mathstrut (i)$};
\end{tikzpicture}
\hspace{1cm}
\begin{tikzpicture}[scale=0.35]
\draw (0,0) --(10,10)--(20,0);
\draw (2,2) --(4,0);
\draw (4,4) --(8,0);
\draw (6,6) --(12,0);
\draw (8,8) --(20,-4)--(14,-4)--(17,-1);

\draw [>=latex, ->, dashed, thick] (0.5,0.5) --(19.5,0.5);
\draw [>=latex, ->, dashed, thick] (3,1) --(19,1);
\draw [>=latex, ->, dashed, thick] (6.5,1.5) --(18.5,1.5);
\draw [>=latex, ->, dashed, thick] (10,2) --(18,2);

\draw[fill] (2,2) circle [radius=1.5pt];
\draw[fill] (4,4) circle [radius=1.5pt];
\draw[fill] (6,6) circle [radius=1.5pt];
\draw[fill] (8,8) circle [radius=1.5pt];
\draw[fill] (0.5,0.5) circle [radius=1.5pt];
\draw[fill] (19.5,0.5) circle [radius=1.5pt];
\draw[fill] (3,1) circle [radius=1.5pt];
\draw[fill] (19,1) circle [radius=1.5pt];
\draw[fill] (6.5,1.5) circle [radius=1.5pt];
\draw[fill] (18.5,1.5) circle [radius=1.5pt];
\draw[fill] (10,2) circle [radius=1.5pt];
\draw[fill] (18,2) circle [radius=1.5pt];

\node at (17,-3) {$\mathstrut T_2$};

\node at (10,-5) {$\mathstrut (ii)$};
\end{tikzpicture}
\caption{(i) A caterpillar network grafted to a tree to form the network $N'$ in Theorem \ref{RootRel}; (ii) The network $N'$ from (i) modified by relocating the root to the arc $a$. In both diagrams a triangle indicates a tree structure.}
\label{RRPic}
\end{figure}

\begin{proof}
Let $N''$ be the network obtained by relocating the root of $N$ to the arc connecting the root of the enhanced caterpillar network to its reticulation descendant, marked $a$ in Figure \ref{RRPic}(i). Then $N''$ is of the form depicted in Figure \ref{RRPic}(ii), and all pairwise distances are retained. This is because for any display tree $\T_i'$ of $N'$, the corresponding display tree $\T_i$ of $N$ is identical when considered as an unrooted tree. It follows that pairwise distances are retained for any pair of leaves $x_i,x_j$ in $N'$.

With the new root location in $N''$ we have a tree-metrizable network --- a caterpillar network --- with a tree leaf-grafted onto it, so we can invoke Theorem \ref{t:LeafReplacement} to observe that $N''$ is tree-metrizable. Since pairwise distances are retained in the transformation from $N'$ to $N''$, it follows that $N'$ is tree-metrizable too.
\end{proof}

\begin{Example}
We now return to considering the examples from Figure \ref{NotTMNetwork}. We can now see that the relevant difference between $N_1$ and $N_2$ is that when we perform the corresponding root relocations, $N_1$ does not become a tree-metrizable network, whereas $N_2$ does.
\end{Example}

\section{Discussion and Future Questions}\label{s:discussion}

In this paper we have explored the observation from~\cite{francis2015tree} that it is possible for a phylogenetic network to carry a tree metric.  This observation means that a tree metric may be consistent with a network (or even many networks), in addition to the unique tree specified by the four point condition (Theorem~\ref{t:4PC}).

Specifically, we have asked which networks could possibly carry a tree metric, and addressed this in two main directions.  
Firstly, we have shown in Section~\ref{s:network.stock} that one can ``grow'' a tree-metrizable network, using leaf-grafting, so that one may construct a tree-metrizable network of any height or number of leaves by grafting a tree onto the leaf of an existing tree-metrizable network.  
Secondly, we have shown that any tree metric at all of height $h$ can be represented on a network with $h-1$ non-trivial reticulations (Corollary~\ref{c:height.h}).  
Furthermore, we have considered the grafting of networks onto trees to obtain tree-metrizable networks, and described a class of networks for which this is always possible (Section~\ref{s:tree.stock}).

There are several possible avenues of research opened by these results. For instance, we showed that there exists tree-metrizable HGT networks on $n$ leaves with $n-1$ non-trivial reticulation arcs. A reasonable question is whether even more complexity can be concealed: do there exist HGT networks with more arcs?

A second question for consideration is whether we can characterise exactly which tree-metrizable HGT networks can graft onto a tree? We have discovered one such class of structures (Theorem~\ref{RootRel}), but are there more? Such a classification would, hopefully, give us insight into what possible structures involving horizontal gene transfer can occur in recent history and still appear tree-like.

We finish the discussion with two conjectures regarding the mathematical structures behind tree-metrizability. Proof in either direction of these conjectures would allow deeper insight into tree-metrizability and may lead to answering our previous questions.

\begin{Conjecture}
Let $N$ be a level-$k$ HGT network with a biconnected component $B$. Then the minimal support network $N(B)$ of $B$ has at most $k-1$ cherries unless $N(B)$ strictly embeds in the top of $N$ or all reticulation arcs in $N(B)$ are between adjacent arcs of $T_{N(B)}$.
\end{Conjecture}

\begin{Conjecture}
Let $N(B)$ be a tree-metrizable minimal support network of a biconnected component $B$. Suppose $N(B)$ has $n$ leaves. Then $N(B)$ has at least $n-2$ reticulation arcs.  
\end{Conjecture}

\appendix

\section{Proofs of Lemma \ref{InequalityLemma} and Lemma~\ref{OmitAdjacent}}

 \begin{Lemma}[Lemma \ref{InequalityLemma}]
 Let $N$ be a four-leaf network with leaves $\{x_1,x_2,x_3,x_4\}$. Suppose $| \T_N|=3,$, so $\T_N = \{T_r, T_s, T_t \}$. Let $T_w(N)=\{T_1^{w_1},\dots,T_k^{w_k} \}$, and suppose that $\alpha_j$ is the probability assigned to $T_j$, and $p_j$ is the length of the internal arc of $T_j$. Then $N$ is tree-metrizable on $T_r$ if and only if
 \begin{center}
 $\displaystyle \sum_{T_j \cong T_r} p_j \alpha_j > \sum_{T_j \cong T_s} p_j \alpha_j = \sum_{T_j \cong T_t} p_j \alpha_j$.
 \end{center}
 \end{Lemma}

 \begin{proof}
 Label the four leaves $\{x_1,x_2,x_3,x_4\}$. 
 Let 
 \begin{center}
 $\displaystyle d_i = \sum_{T_j \cong T_i} \alpha_j d_{(T_j^{w_j})}$
 \end{center}
 for $i=r,s,t$, and let $d=d_1+d_2+d_3$. Let $a_i=d_{(T_i^{w_i})}(x_1,x_2)+d_{(T_i^{w_i})}(x_3,x_4)$, $b_i=d_{(T_i^{w_i})}(x_1,x_3)+d_{(T_i^{w_i})}(x_2,x_4)$ and $c_i=d_{(T_i^{w_i})}(x_1,x_4)+d_{(T_i^{w_i})}(x_2,x_3)$. Define
\begin{align*}
 \setlength\itemsep{0em}
 A_i & = d_i(x_1,x_2)+d_i(x_3,x_4)= \sum_{T_j \cong T_i} a_j \alpha_j \\
 B_i & = d_i(x_1,x_3)+d_i(x_2,x_4)= \sum_{T_j \cong T_i} b_j \alpha_j  \\
 C_i & = d_i(x_1,x_4)+d_i(x_2,x_3)= \sum_{T_j \cong T_i} c_j \alpha_j .
 \end{align*}
Without loss of generality, suppose $T_r$, $T_s$ and $T_t$ are the quartets $x_1x_2|x_3x_4$, $x_1x_3|x_2x_4$ and $x_1x_4|x_2x_3$ respectively. It follows that $A_r < B_r = C_r, B_s < A_s = C_s$, and $C_t < A_t = B_t$. 

 We shall now find a probability distribution so that $d$ is a tree metric on $x_1x_2|x_3x_4$. Let $S_1 = d(x_1,x_2)+d(x_3,x_4), S_2 = d(x_1,x_3)+d(x_2,x_4)$, and $S_3 = d(x_1,x_4)+d(x_2,x_3)$. Then $d$ is a tree metric iff $S_1 < S_2 = S_3$. Considering $S_2=S_3$ first,
\begin{align*}
 &&d(x_1,x_3)+d(x_2,x_4) & = d(x_1,x_4)+d(x_2,x_3) \\
 &\implies &B_r+B_s+B_t & = C_r + C_s + C_t \\
 &\implies &B_t - C_t & = C_s -B_s \\
 &\implies &\sum_{T_j \cong T_t} b_j \alpha_j - \sum_{T_j \cong T_t} c_j \alpha_j & = \sum_{T_j \cong T_s} c_j \alpha_j - \sum_{T_j \cong T_s} b_j \alpha_j \\
 &\implies &\sum_{T_j \cong T_t} (b_j - c_j) \alpha_j & = \sum_{T_j \cong T_s} (c_j - b_j) \alpha_j,
 \end{align*}
which, noting that $b_j-c_j = 2p_j$ for $T_j \cong T_t$ and $c_j-b_j= 2p_j$ for $T_j \cong T_s$ implies that
 \begin{center}
 $\displaystyle \sum_{T_j \cong T_t} p_j \alpha_j = \sum_{T_j \cong T_s} p_j \alpha_j$,
 \end{center}
 which is the desired equality for this lemma.

 Now considering the requirement that $S_1 < S_2$, we have
 \begin{align*}
&& d(x_1,x_2)+d(x_3,x_4) & < d(x_1,x_3) + d(x_2,x_4) \\
&\implies & A_r + A_s + A_t & < B_r + B_s + B_t \\
&\implies & A_s - B_s & < B_r - A_r \\
&\implies & \sum_{T_j \cong T_s} a_j \alpha_j - \sum_{T_j \cong T_s} b_j \alpha_j & < \sum_{T_j \cong T_r} b_j \alpha_j - \sum_{T_j \cong T_r} a_j \alpha_j \\
&\implies & \sum_{T_j \cong T_s} (a_j - b_j) \alpha_j & < \sum_{T_j \cong T_r} (b_j - a_j) \alpha_j
 \end{align*}
 which, noting that $a_j-b_j = 2p_j$ for $T_j \cong T_s$ and $b_j-a_j= 2p_j$ for $T_j \cong T_r$ implies that
\[ \sum_{T_j \cong T_s} p_j \alpha_j < \sum_{T_j \cong T_r} p_j \alpha_j,\]
which is the inequality required for this lemma. The remaining cases (e.g. when $T_r=x_1x_3|x_2x_4, T_s=x_1x_2|x_3x_4$ and $T_t=x_1x_4|x_2x_3$, etc.) are proved similarly.
 \end{proof}

\begin{Lemma}[Lemma \ref{OmitAdjacent}]
Let $N^w$ be an HGT network with reticulation probabilities $\beta$, and an HGT arc $h$ between a pair of siblings such that there is no other vertex between the ends of $a$ and their common parent vertex. Let $N_1$ be the HGT network obtained by deleting $a$ from $N$  and suppressing its end vertices. Then there exist arc weights $w_1$ and reticulation probabilities $\beta_1$ on $N_1$ so that
\begin{center}
$d_{(N_1^{w_1},\beta_1)} = d_{(N^w,\beta)}$.
\end{center}
\end{Lemma}

\begin{proof}
Suppose the weighted display trees of $N^w$ are $\T_N^w=\{T_1^{w_1},\dots,T_{2^r}^{w_{2^r}} \}$ with respective probabilities $\beta_i$ and the weighted display trees of $\widehat{N}$ are $\T_{\widehat{N}}^{\widehat{w}}= \{ \widehat{T}_1^{\widehat{w}_1},\dots,\widehat{T}_{2^{r-1}}^{\widehat{w}_{2^{r-1}}} \}$ with respective probabilities $\widehat{\beta}_i$. 

Consider $\T_{\widehat{N}}^{\widehat{w}}$. We can associate each weighted display tree $\widehat{T}_i^{\widehat{w}_i}$ of $\widehat{N}^{\widehat{w}}$ in a natural way to a pair of weighted display trees of $N^w$, by considering those trees made with the same selection of arcs as $N$ plus either keeping or deleting $h$. Rearrange the indexing of $\T_{N}^{w}$ if necessary so that $\widehat{T}_i^{\widehat{w}_i}$ is associated with $T_{2i-1}^{w_{2i-1}}, T_{2i}^{w_{2i}}$. If we can ensure that
\begin{center}
$\displaystyle \beta_{2i-1} d_{(T_{2i-1}^{w_{2i-1}})} + \beta_{2i} d_{(T_{2i}^{w_{2i}})}= \widehat{\beta}_i d_{(\widehat{T}_i^{\widehat{w}_i})}$
\end{center}
then the lemma is proven.

To this end, set all arc weights and reticulation probabilities to be identical in $N^w$ and $\widehat{N}^{\widehat{w}}$ except for those depicted in Figure \ref{AdjacentRetic}. Label $v_1,v_2,v_3$ as in Figure \ref{AdjacentRetic} such that there are no vertices between them that are not shown in the diagram. Then we label the arc weights and reticulations between $v_1, v_2$ and $v_3$ in $N^w$ and $\widehat{N}^{\widehat{w}}$ as in Figure \ref{AdjacentRetic}.

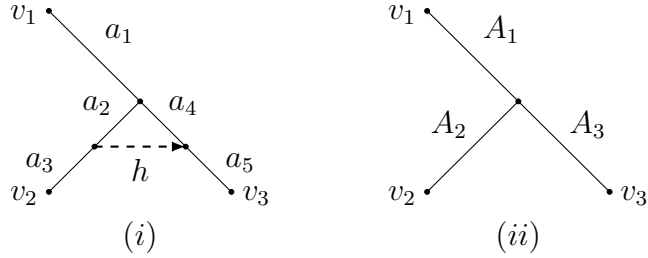
\begin{figure}[H]
\centering
\begin{tikzpicture}[scale=0.6]
\draw (0,0) --(2,2)--(4,0);
\draw (2,2) --(0,4);

\draw [>=latex, ->, dashed, thick] (1,1) --(3,1);

\draw[fill] (0,0) circle [radius=1.5pt];
\draw[fill] (2,2) circle [radius=1.5pt];
\draw[fill] (4,0) circle [radius=1.5pt];
\draw[fill] (0,4) circle [radius=1.5pt];
\draw[fill] (1,1) circle [radius=1.5pt];
\draw[fill] (3,1) circle [radius=1.5pt];

\node[above right=-0.01cm] at (1,3) {$\mathstrut a_1$};
\node[above left=-0.1cm] at (1.5,1.5) {$\mathstrut a_2$};
\node[above left=-0.1cm] at (0.25,0.25) {$\mathstrut a_3$};
\node[above right=-0.1cm] at (2.5,1.5) {$\mathstrut a_4$};
\node[above right=-0.1cm] at (3.75,0.25) {$\mathstrut a_5$};

\node[below] at (2,1) {$\mathstrut h$};

\node[left] at (0,4) {$\mathstrut v_1$};
\node[left] at (0,0) {$\mathstrut v_2$};
\node[right] at (4,0) {$\mathstrut v_3$};

\node at (2,-1) {$\mathstrut (i)$};

\end{tikzpicture}
\hspace{1cm}
\begin{tikzpicture}[scale=0.6]
\draw (0,0) --(2,2)--(4,0);
\draw (2,2) --(0,4);

\draw[fill] (0,0) circle [radius=1.5pt];
\draw[fill] (2,2) circle [radius=1.5pt];
\draw[fill] (4,0) circle [radius=1.5pt];
\draw[fill] (0,4) circle [radius=1.5pt];

\node[above right=-0.01cm] at (1,3) {$\mathstrut A_1$};
\node[above left=-0.1cm] at (1,1) {$\mathstrut A_2$};
\node[above right=-0.1cm] at (3,1) {$\mathstrut A_3$};

\node[left] at (0,4) {$\mathstrut v_1$};
\node[left] at (0,0) {$\mathstrut v_2$};
\node[right] at (4,0) {$\mathstrut v_3$};

\node at (2,-1) {$\mathstrut (ii)$};

\end{tikzpicture}
\caption{(i) The section between vertices $v_1,v_2$ and $v_3$ in $N^w$; (ii) The corresponding section of $\widehat{N}^{\widehat{w}}$.}
\label{AdjacentRetic}
\end{figure}

We now consider the resulting weighted display trees. Of course, $v_2$ and $v_3$ may not appear in $T_{2i-1}$ and $T_{2i}$, but this will occur if and only if those same arcs/vertices do not appear in $\widehat{T}_i$. The corresponding arc weights of the weighted display trees obtained by keeping/deleting $h$ (in the case that $v_2$ and $v_3$ are in the display tree) are shown in Figure \ref{wdt1}.

\begin{figure}[H]
\centering
\begin{tikzpicture}[scale=0.6]
\draw (0,0) --(2,2)--(4,0);
\draw (2,2) --(0,4);

\draw[fill] (0,0) circle [radius=1.5pt];
\draw[fill] (2,2) circle [radius=1.5pt];
\draw[fill] (4,0) circle [radius=1.5pt];
\draw[fill] (0,4) circle [radius=1.5pt];

\node[above right=-0.01cm] at (1,3) {$\mathstrut a_1+a_2$};
\node[above left=-0.1cm] at (1,1) {$\mathstrut a_3$};
\node[above right=-0.1cm] at (3,1) {$\mathstrut a_5$};

\node[left] at (0,4) {$\mathstrut v_1$};
\node[left] at (0,0) {$\mathstrut v_2$};
\node[right] at (4,0) {$\mathstrut v_3$};

\node at (2,-1) {$\mathstrut (i)$};

\end{tikzpicture}
\hspace{1cm}
\begin{tikzpicture}[scale=0.6]
\draw (0,0) --(2,2)--(4,0);
\draw (2,2) --(0,4);

\draw[fill] (0,0) circle [radius=1.5pt];
\draw[fill] (2,2) circle [radius=1.5pt];
\draw[fill] (4,0) circle [radius=1.5pt];
\draw[fill] (0,4) circle [radius=1.5pt];

\node[above right=-0.01cm] at (1,3) {$\mathstrut a_1$};
\node[above left=-0.1cm] at (1,1) {$\mathstrut a_2+a_3$};
\node[above right=-0.1cm] at (3,1) {$\mathstrut a_4+a_5$};

\node[left] at (0,4) {$\mathstrut v_1$};
\node[left] at (0,0) {$\mathstrut v_2$};
\node[right] at (4,0) {$\mathstrut v_3$};

\node at (2,-1) {$\mathstrut (ii)$};

\end{tikzpicture}

\caption{(i) The relevant section of a display tree of $N$ that contains vertices $v_1,v_2,v_3$ when keeping $\alpha$; (ii) The corresponding section with deletion of $\alpha$.}
\label{wdt1}
\end{figure}
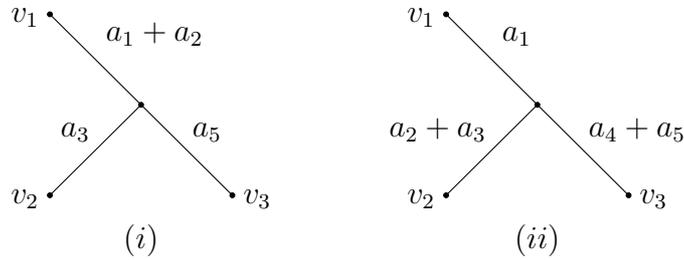

If we can set values for $A_1,A_2,A_3$ so that
\[ d_{(\widehat{T}_i^{\widehat{w}_i},\widehat{\beta}_i)}(s,t) = d_{(T_{2i-1}^{w_{2i-1}},\beta_{2i-1})}(s,t) + d_{(T_{2i}^{w_{2i}},\beta_{2i})}(s,t)
\]
for $(s,t)=(v_1,v_2),(v_1,v_3)$ and $(v_2,v_3)$, then we achieve the required result. In particular, this would follow if these equalities hold:
\begin{align*}
A_1+A_2 & = \alpha(a_1+a_2+a_3) + (1-\alpha)(a_1+a_2+a_3) & (\text{distance }v_1\text{ to }v_2)\\
& = a_1+a_2+a_3 \\
A_1+A_3 & = a_1+a_5+\alpha a_2 + (1-\alpha)a_4 & (\text{distance }v_1\text{ to }v_3)\\
A_2+A_3 & = a_3+a_5+(1-\alpha)(a_2+a_4). & (\text{distance }v_2\text{ to }v_3)
\end{align*}
Setting $A_1=a_1+\alpha a_2$, $A_2=a_3+(1-\alpha)a_2$ and $A_3=a_5+(1-\alpha)a_4$ satisfies these criteria, and all of these values are positive. This completes the result.
\end{proof}

\end{document}